\documentclass[11pt]{article}

\usepackage{amssymb, amsmath, enumerate, paralist}

\usepackage{graphicx,picins}

\usepackage{algorithmicx}
\usepackage[noend]{algpseudocode}
\usepackage{amsmath,amsfonts,mathrsfs}
\usepackage{xspace}
\usepackage{graphicx}
\usepackage{epsfig}
\usepackage{authblk}
\usepackage{algorithm}
\usepackage{makeidx}  

\usepackage{tikz}
\usetikzlibrary{backgrounds}
\usetikzlibrary{arrows}

 \tikzstyle{Gfacility}=[rectangle,draw,fill=black!90,minimum size=4.5pt, inner sep=0pt]
 \tikzstyle{Filtered}=[rectangle,draw,fill=black!20,minimum size=4.5pt, inner sep=0pt]
 \tikzstyle{Lfacility}=[circle,draw=black,fill=black!90,minimum size=4.5pt,inner sep=0pt]
 \tikzstyle{LFiltered}=[circle,draw=black,fill=black!20,minimum size=4.5pt,inner sep=0pt]

 \tikzstyle{Point}=[circle,draw=black,fill=blue!90,minimum size=4pt,inner sep=0pt]

 \tikzstyle{Lfacility-tiny}=[circle,draw=black,fill=black!90,minimum size=1pt,inner sep=0pt]

\newtheorem{theorem}{Theorem} 
\newtheorem{lemma}{Lemma} 
\newtheorem{claim}{Claim}

\newtheorem{definition}{Definition} 
\newtheorem{remarka}{Remark}

\newenvironment{proof}{{\bf Proof.}}{\hfill\rule{2mm}{2mm}} 

\newenvironment{pproof}[1]{\noindent{\textbf{Proof of #1.}}}{\hfill\rule{2mm}{2mm}} 

\newcommand{\RR}{\mathbb{R}}

\newcommand{\cB}{{\mathcal B}}
\newcommand{\cost}{{\rm cost}\xspace}
\newcommand{\optref}{*\xspace}
\newcommand{\locref}{\xspace}
\newcommand{\cX}{{\mathcal X}\xspace}
\newcommand{\fa}{{\mathcal C}\xspace}
\newcommand{\opt}{{\mathcal O}\xspace}
\newcommand{\loc}{{\mathcal S}\xspace}
\newcommand{\fopt}{\overline{\opt}\xspace}
\newcommand{\floc}{\overline{\loc}\xspace}
\newcommand{\costo}{c^{\optref}\xspace}
\newcommand{\costl}{c^{\locref}\xspace}
\newcommand{\swaps}{\rho(\eps, d)\xspace}
\newcommand{\cent}{{\rm cent}\xspace}
\newcommand{\ophi}{{\overline{\phi}}\xspace}
\newcommand{\osigma}{{\overline{\sigma}}\xspace}

\newcommand{\cents}{{\mathcal T}\xspace}
\newcommand{\net}{{\mathcal N}\xspace}
\newcommand{\eps}{{\epsilon}}

\newcommand{\calP}{{\cal P}}

\newcommand{\del}{\delta}

\newcommand{\kmeans}{{$k$-\textsc{means}}\xspace}
\newcommand{\kmed}{{$k$-\textsc{median}}\xspace}
\newcommand{\ufl}{{\textsc{uncapacitated facility location}}\xspace}
\newcommand{\gkm}{{\textsc{generalized} $k$-\textsc{median}}\xspace}
\newcommand{\lqnorm}{{$\ell_q^q$-\textsc{norm~}$k$-\textsc{clustering}}\xspace}

\newcommand{\apx}{{\textbf{APX}}}
\newcommand{\ETH}{{\textbf{ETH}}}

\title{Local Search Yields a PTAS for $k$-Means in Doubling Metrics
\footnote{A preliminary version of this paper appeared in Proceedings of 57th FOCS 2016 \cite{FRS16B}}}
\author{Zachary Friggstad\thanks{This research was undertaken, in part, thanks to funding from the Canada Research Chairs program and an NSERC Discovery Grant.} \qquad Mohsen Rezapour \qquad Mohammad R. Salavatipour\thanks{Supported by NSERC.}\\
Department of Computing Science\\
University of Alberta}

\date{}
\begin{document}

\maketitle

\begin{abstract}
The most well known and ubiquitous clustering problem encountered in nearly every branch of science is undoubtedly \kmeans:
given a set of data points and a parameter $k$, select $k$ centres and partition the data points into $k$ clusters
around these centres so that the sum of squares of distances of the points to their cluster centre 
is minimized. Typically these data points lie in Euclidean space $\RR^d$ for some $d\geq 2$.

\kmeans and the first algorithms for it were introduced in the 1950's.
Over the last six decades, hundreds of papers have studied this problem and different algorithms have been proposed for it.
The most commonly used algorithm in practice  is known as Lloyd-Forgy, which is also referred to as ``the'' \kmeans algorithm, 
and various extensions of it
often work very well in practice.
However, they may produce solutions whose cost is arbitrarily large compared to the optimum solution.
Kanungo et al.~[2004] 
analyzed a very simple local search heuristic
to get a polynomial-time 
algorithm with approximation ratio $9+\epsilon$ for any fixed $\epsilon>0$ for \kmeans in Euclidean space.

Finding an algorithm with a better worst-case approximation guarantee has remained one of the biggest open questions 
in this area, in particular whether one can get a true PTAS for fixed dimension Euclidean space.
We settle this problem by showing that a simple local search algorithm 
provides a PTAS for \kmeans for $\RR^d$ for any fixed $d$.

More precisely, for any error parameter $\epsilon>0$, 
the local search algorithm that considers swaps of up to $\rho=d^{O(d)}\cdot{\epsilon}^{-O(d/\epsilon)}$ centres 
at a time will produce a solution using {\em exactly} $k$ centres
whose cost is at most a $(1+\epsilon)$-factor greater than the optimum solution.
Our analysis extends very easily to the more general settings where we want to minimize the sum of $q$'th powers
of the distances between data points and their cluster centres (instead of sum of squares of distances as in \kmeans) 
for any fixed $q\geq 1$ and where the metric may not be Euclidean but still has fixed doubling dimension.

Finally, our techniques also extend to other classic clustering problems. We provide the first demonstration that local search
yields a PTAS for \ufl and the generalization of \kmed to the setting with non-uniform opening costs in doubling metrics.
\end{abstract}

\section{Introduction}
With advances in obtaining and storing data, one of the emerging challenges of our age is data analysis.
It is hard to find a scientific research project which does not involve some form of methodology to process,
understand, and summarize data. A large portion of data analysis is concerned with predicting patterns in data
after being trained with some training data set (machine learning).
Two problems often encounterd in data analysis are classification and clustering.
Classification (which is an instance of supervised learning) is the task of predicting the label of a new data point
after being trained with a set of labeled data points (called a training set).
Basically, after given a training set of correctly labeled data points the program has to identify the label of
a given new (unlabeled) data point. Clustering (which is an instance of unsupervised learning)
is the task of grouping a given set of objects or data points into clusters/groups
such that the data points that are more similar fall into the same cluster while 
data points (objects) that do not seem similar are in different clusters. 
Some of the main purposes of clustering are to understand the underlying structure and relation between objects
and find a compact representation of data points.

Clustering and different methods to achieve it have been studied since the
1950's in different branches of science: biology, statistics, medical sciences, computer science, social sciences, 
engineering, physics, and more.
Depending on the notion of what defines a cluster, different models of clustering have been proposed and studied by
researchers. Perhaps the most widely used clustering model 
is the \kmeans clustering: Given a set $\cX$ of $n$ {\em data points}
in $d$-dimensional Euclidean space $\RR^d$, and an integer $k$, 
find a set of $k$ points $c_1,\ldots,c_k \in \RR^d$ to act as as {\em centres}
that minimize the sum of squared distances of each data point to its nearest 
centre. In other words, we would like to partition $\cX$ into $k$ cluster sets, $\{C_1,\ldots,C_k\}$ and 
find a centre $c_i$ for each $C_i$ to minimize
\[\sum_{i=1}^k\sum_{x\in C_i}||x-c_i||^2_2. \]
Here, $||x-c_i||_2$ is the standard Euclidean distance in $\RR^d$ between points $x$ and $c_i$.

This value is called the cost of the clustering.
Typically, the centres $c_i$ are selected to be the centroid (mean) of the cluster $C_i$.
In other situations the centres must be from the data points themselves (i.e. $c_i\in C_i$) or from
a given set $\fa$. This latter version is 
referred to as discrete \kmeans clustering. Although in most application of \kmeans the data points are in some
Euclidean space, the discrete variant can be defined in general metrics. 
The \kmeans clustering problem is known to be an NP-hard problem even for $k=2$ or 
when $d=2$~\cite{ADHP09,MNV09,DFKVV04,Vattani09}.


Clustering, in particular the \kmeans clustering problem as the most popular model for it, has found numerous applications in very different areas. The following is a (short) list of applications of clustering that have been addressed by Jain \cite{Jain10}: image segmentation, information access, grouping customers into different types for efficient marketing, grouping delivery services for workforce management and planning, and grouping genome data in biology.
For instance, clustering is used to identify groups of genes with related expression patterns in a key step of the analysis of gene functions and cellular processes.
It is also extensively used to group patients based on their genetic, pathological, and cellular features which is proven useful in analyzing
human genetics diseases (e.g., see \cite{Patrik2005,Hofree2013}).

The most widely used algorithm for \kmeans (which is also sometimes referred to as ``the'' $k$-means algorithm)
is a simple heuristic introduced by Lloyd in 1957~\cite{Lloyd82}. 
This algorithm starts from an initial partition of the points into $k$ clusters 
and it repeats the following two steps as long as it improves the quality of the clustering:
pick the centroids of the clusters as centres, and then re-compute a new clustering by assigning
each point to the nearest centre. Although this algorithm works well in practice it is known that the 
ratio of the cost of the solution computed by this algorithm vs the optimum solution cost 
(known as the ``approximation ratio'') can be  arbitrarily large (see~\cite{KMNPSW04}).
Various modifications and extensions of this algorithm
have been produced and studied, e.g. ISODATA, FORGY, Fuzzy C-means, $k$-means++, 
filtering using kd-trees (see~\cite{Jain10}), but none of them are known to have a bounded approximation ratio
in the general setting. Arthur and Vassilvitskii~\cite{AV07} show that Lloyd's method with properly chosen
initial centres will be an $O(\log k)$-approximation. Ostrovsky et al.~\cite{ORSS12} show that under some assumptions
about the data points the approximation ratio is bounded by a constant.
The problem of finding an efficient algorithm for \kmeans
with a proven theoretical bound on the cost of the solution returned is probably one of the most well studied problems
in the whole field of clustering with hundreds of research papers devoted to this.

Arthur and Vassilvitskii~\cite{AV06}, Dastupta and Gupta~\cite{Dasgupta03a}, and Har-Peled and Sadri
\cite{HS05} study convergence rate of Lloyd's algorithm. In particular~\cite{AV06} show that it can be 
super-polynomial. More recently, Vattani~\cite{Vattani11} shows that it can take exponential time even in two dimensions.
Arthur et al.~\cite{AMR11} proved that Lloyd's algorithm has polynomial-time smoothed complexity.
Kumar and Kannan~\cite{KK10}, Ostrovsky et al.~\cite{ORSS12}, and Awasthi et al.~\cite{ABS10}
gave empirical and theoretical evidence for when and why the known heuristics work well in practice.
For instance~\cite{ABS10} show that when the size of an optimum $(k-1)$-{\textsc{means}} is sufficiently larger than the
cost of \kmeans then one can get a near optimum solution to \kmeans using a variant of Lloyd's algorithm.

The \kmeans problem is known to be NP-hard~\cite{ADHP09,DFKVV04,MNV09}.
In fact, the \kmeans problem is NP-hard if $d$ is arbitrary even for $k=2$~\cite{ADHP09,DFKVV04}. Also, if $k$ is arbitrary
the problem is NP-hard even for $d=2$~\cite{MNV09,Vattani09}.
However, the \kmeans problem can be solved in polynomial time by the algorithm of~\cite{Inaba1994} when both $k$ and $d$ are constant.

A {\em polynomial-time approximation scheme} (PTAS) is an algorithm that accepts an additional parameter $\eps > 0$.
It finds solutions whose cost is at most $1 + \eps$ times the optimum solution cost
and runs in polynomial time when $\eps$ can be regarded as a fixed constant (i.e. $O(n^{f(\eps)})$ for some function $f$).
Matou\v{s}ek~\cite{Matousek00} gave a PTAS for fixed $k,d$ with running time $O(n (\log n)^k\eps^{-2k^2d})$.
Since then several other PTASs have been
proposed for variant settings of parameters but all need $k$ to be constant 
\cite{OR02,BHI02,DKKR03,HM04,KSS04,KSS10,FMS07,HK05}.
B\={a}doiu et al.~\cite{BHI02} gave a PTAS for fixed $k$ and any $d$ with time 
$O(2^{(k/\epsilon)^{O(1)}}{\rm poly}(d)n\log^k n)$.
De la Vega et al.~\cite{DKKR03} proposed a $(1+\epsilon)$-approximation
with running time $O(2^{(k^3/\epsilon^8)(\ln(k/\epsilon))\ln k}dn\log^k n)$. This was improved to 
$O(2^{(k/\epsilon)^{O(1)}}dn)$ by Kumar et al.~\cite{KSS04,KSS10}.
By building {\em coresets} of size $O(k\epsilon^{-d}\log n)$,
Har-Peled and Mazumdar~\cite{HM04} presented a $(1+\epsilon)$-approximation with time 
$O(n+k^{k+1}\epsilon^{-(2d+1)k}\log^{k+1} n\log^k\frac{1}{\epsilon})$. This was slightly improved by Har-Peled and 
Kushal~\cite{HK05}. Feldman et al.~\cite{FMS07} developed another $(1+\epsilon)$-approximation with running time 
$\tilde{O}(nkd+d\cdot {\rm poly}(k/\epsilon)+2^{\tilde{O}(k/\epsilon)})$, where the $\tilde{O}(.)$ notation hides polylog factors.
Recently, Bandyapadhyay and Varadarajan~\cite{BV16} presented a pseudo-approximation for \kmeans in fixed-dimensional Euclidean space:
their algorithm finds a solution whose cost is at most $1+\epsilon$ times of the optimum but might use up to
$(1+\epsilon)\cdot k$ clusters.

The result of Matou\v{s}ek~\cite{Matousek00} also shows that 
one can select a set $\fa$ of ``candidate'' centers in $\RR^d$ from which the $k$ centres should be chosen from with a loss
of at most $(1+\epsilon)$ and this set can be computed in time $O(n\epsilon^{-d}\log(1/\epsilon))$.
This reduces the $k$-means problem to the discrete setting where along with $\cX$ we have a set $\fa$ of candidate
centres and we have to select $k$ centres from $\fa$. 
Kanungo et al.~\cite{KMNPSW04} proved that a simple local search heuristic yields an algorithm with
approximation ratio $9+\epsilon$ for $\RR^d$. This remains the best known approximation algorithm with polynomial running time
for $\RR^d$. They also present an algorithm which is a hybrid of local
search and Lloyd algorithm and give empirical evidence that this works much better in practice (better than Lloyd's 
algorithm) and has proven constant factor ratio. For general metrics, Gupta and Tangwongsan \cite{GT08} proved that local
search is a $(25+\eps)$-approximation.
It was an open problem for a long time whether \kmeans is APX-hard or not in Euclidean metrics.
This was recently answered positively by Awasthi et al.~\cite{ACKS15} where they showed that the problem is
 APX-hard in $\RR^d$ but the dimension $d$ used in the proof is $\Omega(\log n)$.
Blomer et al.~\cite{BLSS16} have a nice survey of theoretical analysis of different $k$-means algorithms.


~

\noindent
{\bf \kmed}\\
Another very well studied problem that is also closely related to \kmeans is \kmed.
The only difference is that the goal (objective function) in \kmed is to minimize the sum of distances,
instead of sum of square of distances as in \kmeans, i.e. minimize $\sum_{i=1}^k\sum_{x\in C_i} \del(x,c_i)$
where $\del(x,c_i)$ is the distance between $x$ and $c_i$. This problem occurs in operations research settings.

There are constant factor approximation algorithms for
\kmed in general metrics. The simple local search (which swaps in and out a constant number of centres in
each iteration) is known to give a $3+\epsilon$ approximation by Arya et al.~\cite{AGKMMP01,AGKMMP04}. 
The current best approximation uses different techniques and has an approximation ratio
of $2.611+\epsilon$~\cite{LS13,BPRST15}. The local search $(9+\epsilon)$-approximation (for \kmeans) in~\cite{KMNPSW04} 
can be seen as an extension of the analysis in~\cite{AGKMMP01} for \kmed. One reason that analysis of \kmeans
is more difficult is that the squares of distances do not necessarily satisfy the triangle inequality.
For instances of \kmed on Euclidean metrics, Arora et al.~\cite{ARR98},
building on the framework of Arora~\cite{Arora98}, gave the first PTAS.
Kolliopoulos and Rao~\cite{KR99} improved the time complexity and presented a PTAS for Euclidean \kmed with
time complexity $O(2^{O((1+\log\frac{1}{\epsilon})/\epsilon)^{d-1}}n\log n\log k)$. Such an approximation is also known
as an {\em efficient PTAS}: the running time of the $(1+\eps)$-approximation is of the form $f(\eps) \cdot {\rm poly}(n)$
(in fixed-dimension metrics).


~

\noindent
{\bf \textsc{Uncapacitated facility location}}\\
The \ufl problem is the same as \kmed except instead of a cardinality constraint bounding the number of open facilities, we are instead
given opening costs $f_i$ for each $i \in \fa$. The goal is to find a set of centers $\loc \subseteq \fa$ that 
minimizes $\sum_{x \in \cX} \del(x, \loc) + \sum_{i \in \loc} f_i$.
Currently the best approximation for \ufl in general metrics is a 1.488-approximation~\cite{L11}. As with \kmed, a PTAS is known  for \ufl in constant-dimensional Euclidean metrics \cite{ARR98,KR99},
with the latter giving an efficient PTAS.

\subsection{Our result and technique}
Although a  PTAS for \kmed in fixed dimension Euclidean space 
has been known for almost two decades, getting a PTAS for \kmeans in fixed dimension 
Euclidean space has remained an open problem. We provide a PTAS for this setting.
We focus on the discrete case where we have to select a set of $k$ centres from a given set $\fa$.
Our main result is to show that a simple local search heuristic that swaps up 
to $\rho=d^{O(d)}\cdot{\epsilon}^{-O(d/\epsilon)}$ 
centres at a time and assigns each point to the nearest centre is a PTAS for \kmeans in $\RR^d$.

A precise description of the algorithm is given in Section \ref{sec:prem}. At a high level, we start with
any set of $k$ centres $\loc \subseteq \fa$. Then, while there is some other set of $k$ centres $\loc' \subseteq \fa$ with $|\loc - \loc'| \leq \rho$
for some constant $\rho$ such that $\loc'$ is a cheaper solution than $\loc$, we set $\loc \leftarrow \loc'$. Repeat until $\loc$ cannot be improved any further.
Each iteration takes $|\fa|^{O(\rho)}$ time, which is 
 polynomial when $\rho$ is a constant.
Such a solution is called a {\em local optimum} solution with respect to the $\rho$-swap heuristic.
We still have to ensure that the algorithm only iterates a polynomial number of times; a standard modification
discussed in Section \ref{sec:prem} ensures this.

Recall that the {\em doubling dimension} of a metric space is the smallest $\tau$ such that any ball of
radius $2r$ around a point can be covered by at most $2^\tau$ balls of radius $r$.
If the doubling dimension can be regarded as a constant then we call the metric a {\em doubling metric}
(as in~\cite{Talwar04}). The Euclidean metric over $\RR^d$
has doubling dimension $O(d)$. 
Our analysis implies a PTAS for more general settings where the data points are in a metric space
with constant doubling dimension (described below)
and when the objective function is to minimize the sum of $q$'th power of the distances for some fixed $q \geq 1$.

Let $\swaps := d^{O(d)} \cdot \eps^{O(-d/\eps)}$.
We will articulate the absolute constants suppressed by the $O(\cdot)$ notation later on in our analysis.
\begin{theorem}\label{theo:main}
The local search algorithm that swaps up to $\swaps$ centres at a time
is a $(1+\epsilon)$-approximation for \kmeans in metrics with doubling dimension $d$.
\end{theorem}

Now consider the generalization where the objective function measures the sum of $q$'th power of the distances
between points and their assigned cluster centre, we call this \lqnorm. 
Here, we are given the points $\cX$ in a metric space $\del(\cdot,\cdot)$ along with a set $\fa$ of potential centres.
We are to select $k$ centres from $\fa$ and partition the
points into $k$ cluster sets $C_1,\ldots,C_k$ with each $C_i$ having a centre $c_i$ so as to minimize
\[\sum_{i=1}^k \sum_{x\in C_i} \del(x,c_i)^q.\]
Note that the case of $q=2$ is the \kmeans problem and $q=1$ is \kmed.
We note that our analysis extends to provide a PTAS for this setting when $q$ is fixed.
That is, Theorem \ref{theo:main} holds for \lqnorm 
for constants $q \geq 1$, except that
we require that the local search procedure be the $\rho'$-swap 
heuristic where $\rho' = d^{O(d)} \cdot (2^q/\eps)^{O(2^q \cdot d/\eps)}$.

Note that even for the case of \kmed, this is the first PTAS for metrics with constant doubling dimension.
Also, while a PTAS was known for \kmed for constant-dimensional Euclidean metrics, determining if local
search provided such a PTAS was an open problem. For example,~\cite{CM15} shows that local search
can be used to get a $1+\eps$ approximation for \kmed that uses up to $(1 + \eps) \cdot k$ centres.

For the ease of exposition, we initially present the proof restricted to \kmeans in $\RR^d$.
We explain in Section \ref{sec:doubling1} how to extend the analysis to doubling metrics and
describe in Section \ref{sec:lpq} how this can be easily extended to prove Theorem \ref{theo:main}
when the objective is \lqnorm. 

As mentioned earlier, Awasthi et al.~\cite{ACKS15} proved that \kmeans is APX-hard for $d=\Omega(\log n)$
and they left the approximability of \kmeans for lower dimensions as an open problem.
A consequence of our algorithm is that one can get a $(1+\epsilon)$-approximation for \kmeans that runs
in sub-exponential time for values of $d$ up to $O(\log n/\log\log n)$.
More specifically, for any given
$0<\epsilon<1$ and  $d=\sigma\log n/\log\log n$ for sufficiently small absolute constant $\sigma$
we get a $(1+\epsilon)$-approximation for \kmeans that
runs in time $O(2^{n^{\kappa}})$, for some constant $\kappa=\kappa(\sigma) < 1$; for $d=O(\log\log n/\log\log\log n)$ we
get a {\em quasi-polytime approximation scheme} (QPTAS). Therefore, our result in a sense shows that the requirement of~\cite{ACKS15} of $d=\Omega(\log n)$
to prove APX-hardness of \kmeans is almost tight unless ${\rm NP}\subseteq DTIME(2^{n^\sigma})$.

The notion of coresets and using them for finding faster algorithms for $k$-means has been studied extensively
(e.g.~\cite{HM04,HK05,Chen09, FMS07,FSS13} and references there).
A coreset is a small subset of data points (possibly with weights associated to them) such that running the clustering
algorithm on them (instead of the whole data set) generates a clustering of the whole data set with approximately good
cost. In order to do this, one can go to the discrete case.
To use coresets, we need to be able to solve (discrete) $k$-means in the more general
setting where each centre  $i\in\fa$ has an associated weight $w(i)$, and the cost of assigning a point $j$ to $i$ (if $i$
is selected to be a centre) is $w(i)\cdot \del(i,j)^2$.
Our local search algorithm works for this weighted setting as well.
We will show how to use these ideas to improve the running time of the local search algorithm.

Finally, we observe that our techniques easily extend to show a natural local-search heuristic for \ufl is a PTAS.
In particular, for the same constant $\rho$ as in Theorem \ref{theo:main} (in fact, it can be slightly smaller)
we consider the natural local search algorithm that returns a solution $\loc \subseteq \fa$ such that ${\rm cost}(\loc) \leq {\rm cost}(\loc')$
for all $\loc' \subseteq \fa$ with both $|\loc - \loc'| \leq \rho$ and $|\loc' - \loc| \leq \rho$ (i.e. add and/or drop up to $\rho$ centres in $\fa$).
\begin{theorem}\label{thm:ufl}
The local search algorithm that adds and/or drops up to $\swaps$ centres at a time
is a $(1+\epsilon)$-approximation for \ufl in metrics with doubling dimension $d$.
\end{theorem}

This seems to be the first explicit record of a PTAS for \ufl in doubling metrics, but one was known in constant-dimensional Euclidean metrics
\cite{ARR98,KR99}. Prior to this work, local search was only known to provide a PTAS for \ufl in constant-dimensional Euclidean metrics
if all opening costs are the same \cite{CMSOGC15}. The basic idea why this works is our analysis in Theorem \ref{theo:main} considers
test swaps that, overall, swap in each local optimum centre exactly once and swap out each global optimum centre exactly once (i.e. the swaps
come from a full partitioning of the local and global optimum). The partitioning we use to obtain these swaps only uses the assumption that
precisely $k$ centres are open in a feasible solution at one point, and this step can be safely ignored in the case of \ufl.

Lastly, we consider the common generalization of \kmed and \ufl where all centres have costs and at most $k$ centres may be chosen
(sometimes called the \gkm or $k$-\textsc{uncapacitated facility location} problem).
We consider the local search operation that tries to add and/or drop up to $\rho$ centres at a time as long as the candidate solution being tested
includes at most $k$ centres. We get a PTAS in this case as well. To the best of our knowledge, the previous best approximation was a 5-approximation
in general metrics, also obtained by local search \cite{GT08}.
\begin{theorem}\label{thm:gkm}
The local search algorithm that adds and/or drops up to $\swaps$ centres at a time (provided the resulting solution has at most $k$ facilities)
is a a $(1+\epsilon)$-approximation for \gkm in metrics with doubling dimension $d$.
\end{theorem}

The results of Theorems \ref{thm:ufl} and \ref{thm:gkm} extend to the setting where we are considering the $\ell^q_q$-norm
for distances instead of the $\ell_1$-norm.

{\bf Note:} Shortly after we announced our result \cite{FRS16}, Cohen-Addad, Klein, and Mathieu \cite{AKM16,AKM16B} announced
similar results for \kmeans on Euclidean and minor-free metrics using the local search method. Specifically, 
they prove that the same local search algorithm yields a PTAS for \kmeans on Euclidean and minor-free metrics. These results are
obtained independently.

\subsection{Proof Outline} \label{sec:outline}
The general framework for analysis of local search algorithms for \kmed and \kmeans
in~\cite{AGKMMP04,KMNPSW04,GT08} is as follows. Let $\loc$ and $\opt$ be a local optimum and a global optimum solution,
respectively.
They carefully identify a set $Q$ of potential swaps 
between local and global optimum. 
In each such swap, the cost of assigning a data point $x$ to the nearest 
centre after a swap is bounded with respect to the local and global cost assignment. 
In other words, if $\cB(\loc)$ is the set of solutions obtained by performing swaps from $Q$, the main
task is to show that $\sum_{\loc'\in\cB(\loc)} (\cost(\loc')-\cost(\loc))\leq \alpha\cdot \cost(\opt)-\cost(\loc)$ for
some constant $\alpha$. Given that $0 \leq \cost(\loc')-\cost(\loc)$ for all $\loc'\in\cB(\loc)$ (because $\loc$ is a local optimum),
$\cost(\loc)\leq\alpha\cdot \cost(\opt)$.

Our analysis has the same structure but has many more ingredients and several intermediate steps to get us what we want.
Note that the following only describes steps used in the analysis of the local search algorithm; we do not perform any of the
steps described below in the algorithm itself.

Let us define $\loc$ and $\opt$ as before. First, we do a filtering over $\loc$ and $\opt$ to obtain subsets
$\floc\subseteq \loc$ and $\fopt\subseteq\opt$ such that every centre in $\loc-\floc$ (in $\opt-\fopt$)
is ``close'' to a centre in $\floc$ (in $\fopt$) while these filtered centres are far apart. We define a ``net''
around each centre $i\in \floc$ which captures a collection of other filtered centres in $\fopt$ that are relatively
close to $i$. The idea of the net is that if we choose to close $i$ (in a test swap) 
then the data points that were to be assigned to $i$ will be assigned to a nearby centre in the net of $i$. 
Since the metric is a constant-dimensional Euclidean metric (or, more generally, a doubling metric), we can
choose these nets to have constant size.

For each $j$ assigned to $i$ in $\loc$, if the centre $i^*$ that $j$ is assigned to in the optimum solution lies somewhat close to $i$ then we can reassign $j$
to a facility in the net around $i$ that is close to $i^*$. In this case, the reassignment cost for $j$ will be close to $c^*_j - c_j$. Otherwise, if $i^*$ lies far from $i$
then we can reassign $j$ to a facility near $i$ in the net around $i$ and the reassignment cost will only be $O(\eps) \cdot (c^*_j + c_j)$ and we will generate the $c^*_j - c_j$ term
for the local search analysis when $i^*$ is opened in another different swap.

Of course, there are some complications in that we need to do something else with $j$
if the net around $i$ is not open. This will happen infrequently, but we need a somewhat reasonable bound when it does happen.
Also, for reasons that will become apparent in the analysis
we do something different in the case that $i^*$ is somewhat close to $i$ but $i^*$ is much closer to a different facility in $\loc$ than it is to $i$.

 
One main part of our proof is to show that there exists a suitable 
randomized partitioning of $\loc\cup\opt$ such that each
part has small size (which will be a function of only $\epsilon$ and $d$ and, ultimately, determines the size of the swaps in our local search procedure),
and for any pair $(i, i^*) \in \floc \times \fopt$ where $i^*$ lies in the net of $i$
we have $\Pr[i, i^* {\rm ~lie~in~the~same~part}] \geq 1-\eps$.
This randomized partitioning is the only part of the proof that we rely on properties of doubling metrics (or $\RR^d$).
For those small portion of facilities that their net is ``cut'' by our partitioning, we show that when we close
those centres in our test swaps then the reassignment cost of a point $j$ that was assigned to them is only $O(1)$
times more than the sum of their assignment costs in $\opt$ and $\loc$. Given that this only happens with small probability (due to our random
partition scheme), this term is negligible in the final analysis.
So we get an overall bound of $1+O(\epsilon)$ on the ratio of cost of $\loc$ over $\opt$.

{\bf Outline of the paper:} We start with some basic definitions and notation in Section \ref{sec:prem}.
In Sections \ref{sec:euclid} and \ref{sec:partitioning} we show that the local search algorithm with an appropriate number
of swaps provides a PTAS for \kmeans in $\RR^d$.
In Section \ref{sec:doubling} we show how this can be extended to prove Theorem \ref{theo:main}
and to the setting where we measure the $\ell_q^q$-norm of the solution for any constant $q \geq 1$.
Finally, before concluding, in Section \ref{sec:ufl} we show that our analysis easily extends to prove Theorems \ref{thm:ufl} and \ref{thm:gkm}.

\section{Notation and Preliminaries}\label{sec:prem}
Recall that in the \kmeans problem we are given a set $\cX$ of $n$ points in $\RR^d$
and an integer $k\geq 1$; we have to find $k$ centres $c_1,\ldots,c_k \in \RR^d$ so as to minimize the
sum of squares of distances of each point to the nearest centre. As mentioned earlier,
by using the result of~\cite{Matousek00}, at a loss of $(1+\epsilon)$ factor
we can assume we have a set $\fa$ of ``candidate'' centres from which the centres can be chosen from.
This set can be computed in time $O(n\epsilon^{-d}\log(1/\epsilon))$ and $|\fa|=O(n\epsilon^{-d}\log(1/\epsilon))$.
Therefore, we can reduce the problem to the discrete case.

Formally, suppose we are given a set $\fa$ of points (as possible
cluster centres) along with $\cX$ and we have to select the $k$ centres from $\fa$.
Furthermore, we assume the points are given in a metric space $(V,\del)$ (not necessarily $\RR^d$).
For any two points $p,q\in V$, $\del(p,q)$ denotes the distance between them: for the case of the metric being
$\RR^d$, then $\del(p,q)=\sqrt{\sum_{\ell=1}^d |p_\ell-q_\ell|^2}$.

We usually refer to a potential centre in $\fa$ by a simple index $i$ and a point in $\cX$ by a simple index $j$ (or slight
variants like $i^*$ or $\overline{i'}$). This is to emphasize that we do not need to talk about specific coordinates
of points in Euclidean space.
In fact, only once in our proof do we rely on the particular embedding of the points
in Euclidean space. This argument will also be replaced by a more general argument when discussing doubling metrics
in Section \ref{sec:doubling1}.
So, for any set $S\subseteq\fa$ and any $j\in\cX$, let $\del(j,S)=\min_{i\in S}\del(j,i)$.
We also define  $\cost(S) = \sum_{j \in \cX} \del(j, S)^2$.

Our goal in (discrete) \kmeans is to find a set of centres $S\subseteq\fa$ of size $k$ to minimize $\cost(S)$. 
Note that once we fix the set of centres we can find a partitioning of $\cX$ that realizes $\sum_{j\in \cX} \del(j,S)^2$
by assigning each $j\in\cX$ to the nearest centre in $S$, breaking ties arbitrarily.

For ease of exposition we focus on \kmeans on $\RR^d$ and then show
how the analysis can be extended to work for
metrics with constant doubling dimension and when we want to minimize $\sum_{j\in \cX} \del(j,S)^q$ for a fixed $q\geq 1$ instead of just $q=2$.

The simple $\rho$-swap local search heuristic shown in Algorithm \ref{alg:local} 
is essentially the same one considered in~\cite{KMNPSW04}.

\begin{algorithm*}[t]
 \caption{$\rho$-Swap Local Search} \label{alg:local}
\begin{algorithmic}
\State Let $\loc$ be an arbitrary set of $k$ centres from $\fa$
\While{$\exists$ sets $P\subseteq\fa-\loc$, $Q\subseteq \loc$ with $|P|=|Q|\leq \rho$ s.t. 
$\cost((\loc - Q) \cup P)<\cost(\loc)$}
\State $\loc\leftarrow (\loc - Q) \cup P$
\EndWhile
\State \Return $\loc$
\end{algorithmic}
\end{algorithm*}

Recall that we defined $\swaps= d^{O(d)}\cdot\eps^{O(d/\eps)}$, where the constants will be specified later
and consider the local search algorithm with $\rho=\swaps$ swaps. 
By a standard argument (as in~\cite{AGKMMP01,KMNPSW04}) one can show that 
replacing the condition of the while loop with $\cost((\loc - Q) \cup P)\leq(1-\frac{\eps}{k}) \cdot \cost(\loc)$,
the algorithm terminates in polynomial time.
Furthermore, if $\alpha$ is such that any locally optimum solution
returned by Algorithm \ref{alg:local} has cost at most $\alpha \cdot \cost(\opt)$ where $\opt$ denotes a global
optimum solution, then any $\loc$ such that $\cost((\loc - Q) \cup P)<(1-\frac{\eps}{k}) \cdot \cost(\loc)$
for any possible swap $P,Q$ satisfies $\cost(\loc) \leq \frac{\alpha}{1-\eps} \cost(\opt)$.
This follows by arguments in~\cite{AGKMMP01,KMNPSW04} and the fact that our local search analysis
uses at most $k$ ``test swaps''.

For ease of exposition, we ignore this factor $1+\epsilon$ loss, and 
consider the solution $\loc$ returned by Algorithm \ref{alg:local}.
Recall that we use $\opt$ to denote the global optimum solution.
For $j \in \cX$, let $\costo_j = \del(j, \opt)^2$ and $\costl_j = \del(j, \loc)^2$, so $\cost(\opt) = \sum_{j \in \cX} \costo_j$ and $\cost(\loc) = \sum_{j \in \cX} \costl_j$.
We also denote the centre in $\opt$ nearest to $j$ by
$\sigma^*(j)$ and the centre in $\loc$ nearest to $j$ by $\sigma(j)$.
Define $\phi : \opt \cup \loc \rightarrow \opt \cup \loc$ to be the function that assigns 
$i^* \in \opt$ to its nearest centre in $\loc$ and assigns $i\in\loc$ to its nearest centre in $\opt$.
For any two sets $S, T \subseteq \opt \cup \loc$, we let 
$S \triangle T = (S \cup T) - (S \cap T)$.

We assume $\opt \cap \loc = \emptyset$. This is without loss of generality because we could duplicate each 
location in $\fa$ and say $\opt$ uses the originals and $\loc$ the duplicates. It is easy to check that $\loc$ 
would still be a locally optimum solution in this instance. We can also assume that these are the only possible 
colocated facilities, so $\del(i,i') > 0$ for distinct $i,i' \in \opt$ or distinct $i,i' \in \loc$.
Finally, we will assume $\eps$ is sufficiently small (independent of all other parameters, including $d$)  so that all of our bounds hold.


\section{Local Search Analysis for $\RR^d$}\label{sec:euclid}
In this section we focus on $\RR^d$ (for fixed $d\geq 2$) and define $\swaps = 32 \cdot (2d)^{8d} \cdot \eps^{-36 \cdot d/\eps}$.
Our goal in this section is to prove that the $\rho$-swap local search with 
$\rho=\swaps$  is a PTAS for \kmeans in $\RR^d$.

\begin{theorem}\label{thm:euclid}
Let $\loc$ be a locally-optimum solution with respect to the $\swaps$-swap local search heuristic when the points lie in $\RR^d$.
Then $\cost(\loc) \leq (1 + O(\eps)) \cdot \cost(\opt)$.
\end{theorem}

To prove this, we will construct a set of test swaps that yield various inequalities which, when combined, 
provide the desired bound on $\cost(\loc)$. That is, we will partition $\opt \cup \loc$ into sets 
where $|P \cap \opt| = |P \cap \loc| \leq \swaps$ for each part $P$.
For each such set $P$, $0 \leq \cost(\loc \triangle P) - \cost(\loc)$ because $\loc$ is a locally optimum solution. 
We will provide an explicit upper bound on this cost change that will reveal enough information to easily conclude 
$\cost(\loc) \leq (1 + O(\eps)) \cdot \cost(\opt)$.
For example, for a point $j \in \cX$ if $\sigma^*(j) \in P$ then the change in $j$'s assignment cost
is at most $c^*_j - c_j$ because we could assign $j$ from $\sigma(j)$ to $\sigma^*(j)$. The problem is that
points $j$ with $\sigma(j) \in P$ but $\sigma^*(j) \not\in P$ must go somewhere else; most of our effort is ensuring
that the test swaps are carefully chosen so such reassignment cost increases are very small.

First we need to describe the partition of $\opt \cup \loc$. This is a fairly elaborate scheme that involves several
steps. As mentioned earlier, the actual algorithm for \kmeans is the simple local search we described and the
algorithms we describe below to get this partitioning scheme are only for the purpose of proof and analysis of the local search
algorithm.

\begin{definition}
For $i^* \in \opt$ let $D_{i^*} := \del(i^*, \loc) = \del(i^*, \phi(i^*))$. For $i \in \loc$ let $D_i := \del(i, \opt) = \del(i, \phi(i))$.
\end{definition}
The first thing is to sparsify $\opt$ 
and $\loc$ using a simple filtering step. 
Algorithm \ref{alg:filter} {\em filters} $\opt$ to a set that is appropriately sparse for our analysis.

\begin{algorithm*}[t]
\caption{Filtering $\opt$}
\label{alg:filter}
\begin{algorithmic}
\State $\fopt \leftarrow \emptyset$
\For{each $i^* \in \opt$ in nondecreasing order of $D_{i^*}$}
\If {$\exists ~\overline{i^*}\in\fopt$ such that $\del(i^*, \overline{i^*}) \leq \eps \cdot D_{i^*}$}
\State $\eta(i^*) \leftarrow \overline{i^*}$
\Else
\State $\eta(i^*) \leftarrow i^*$
\State $\fopt \leftarrow \fopt \cup \{i^*\}$
\EndIf
\EndFor
\State \Return $\fopt$
\end{algorithmic}
\end{algorithm*}
Think of $\eta(i^*)$ as a proxy for $i^*$ that is very close to $i^*$. 
Using a similar process, we filter $\loc$ to get $\floc$ and proxy centres $\eta(i) \in \floc$ for each $i \in \loc$.
The idea is that the set of centres left in $\fopt$ and $\floc$ are somewhat far apart yet any point that was assigned to a
centre in $\opt-\fopt$ (or in $\loc-\floc$) can be ``cheaply'' reassigned to a proxy.



\begin{lemma}\label{lem:filter}
For each $i \in \opt \cup \loc$ we have $\del(i, \eta(i)) \leq \eps \cdot D_i$. For any distinct 
$i, i' \in \fopt \cup \floc$ we have $\del(i, i') \geq \eps \cdot \max\{D_{i}, D_{i'}\}$.
\end{lemma}
\begin{proof}
That $\del(i, \eta(i)) \leq \eps \cdot D_i$ follows immediately by construction. If $i \in \fopt, i' \in \floc$ or 
vice-versa, then in fact $\del(i, i') \geq \max\{D_i, D_{i'}\}$ simply by definition of $D_i, D_{i'}$.

Now suppose $i, i' \in \fopt$ and that $i'$ was considered after $i$ in Algorithm \ref{alg:filter} (so $D_{i'} \geq D_i$).
The fact that $i'$ was added to $\fopt$ even though $i$ was already in $\fopt$ means
$\del(i, i') \geq \eps \cdot D_{i'}$. The same argument works if $i, i' \in \floc$.
\end{proof}

Next we define mappings similar to $\phi, \sigma, \sigma^*$ except they only concern centres that were not filtered out.\\

\begin{definition}
~
\begin{itemize}
\item $\ophi : \fopt \cup \floc \rightarrow \fopt \cup \floc$ maps each $i \in \fopt$ to its nearest location in 
$\floc$ and vice versa.
\item $\osigma^* : \cX \rightarrow \fopt$ defined by $\osigma^*(j) = \eta(\sigma^*(j))$.
\item $\osigma : \cX \rightarrow \floc$ defined by $\osigma(j) = \eta(\sigma(j))$.
\end{itemize}
Finally, for each $i \in \ophi(\fopt)$, let $\cent(i)$ be the centre in $\ophi^{-1}(i)$ that is closest to $i$, breaking ties arbitrarily.
\end{definition}
Note $\osigma(j)$ may not necessarily be the centre in $\floc$ that is closest to $j$.
Also note that if one considers a bipartite graph with parts $\fopt$ and $\floc$, then $\ophi$ maps centres from one side to the 
other.


\begin{lemma} \label{lem:phi}
For each $i' \in \fopt \cup \floc$, $D_{i'} \leq \del(i', \ophi(i')) \leq (1 + \eps) \cdot D_{i'}$.
\end{lemma}
\begin{proof}
Suppose $i' \in \fopt$, the proof is essentially the same for $i' \in \floc$.
On one hand, we know
\[ D_{i'} = \del(i', \phi(i')) \leq \del(i', \ophi(i')) \]
because $\floc \subseteq \loc$.
On the other hand,
\[ \del(i', \ophi(i')) = \del(i', \eta(\phi(i'))) \leq \del(i', \phi(i')) + \del(\phi(i'), \eta(\phi(i'))) \leq D_{i'} + \eps \cdot D_{\phi(i')}. \]
Conclude by observing $D_{\phi(i')} \leq \del(i', \phi(i')) = D_{i'}$.
\end{proof}

\begin{figure}[t]
\begin{center}
\includegraphics[scale=1.4]{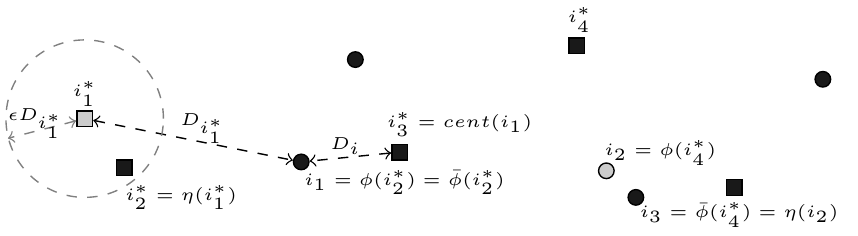} 
\end{center}
 \caption{The square centres lie in $\opt$ and the circle centres lie in $\loc$. The black ones survived the filtering (i.e. lie in $\fopt \cup \floc$) and the grey ones were filtered out.
Note that the grey square $i^*_1$ on the left was too close to $i^*_2$ so it was not added to $\fopt$ and $\eta(i^*_1) = i^*_2$.
Also note that $\ophi(i^*_4) \neq \phi(i^*_4)$ because the facility $i_2$ that defined $D_{i^*_4} := \del(i^*_4, i_2)$ was not added to $\floc$. Still,
$\del(i^*_4, \ophi(i^*_4)) = \del(i^*_4, i_3) \leq (1 + \eps) D_{i_4^*}.$}
 \label{fig:filtering}
\end{figure}

Figure \ref{fig:filtering} depicts many of the concepts covered above.
Finally, the last definition in this section identifies pairs of centres that we would like to have in the same part of the partition we construct.
\begin{definition}
~
\begin{itemize}
\item $\cents := \{(\cent(i), i) : i \in \ophi(\fopt) {\rm ~and~}  \eps \cdot \del(\cent(i), i) \leq D_{i}\}$
\item $\net := \{(i^*, i)  \in \fopt \times \floc :  \del(i, i^*) \leq \eps^{-1} \cdot D_i {\rm ~and~} D_{i^*}  \geq \eps \cdot D_i\}$
\end{itemize}
\end{definition}
For each $i \in \floc$, the set $\{i^* : (i^*, i) \in \net\}$ is the ``net'' for centre $i$ that was discussed
in the proof outline in Section \ref{sec:outline}.

Ultimately we will require that pairs in $\cents$ are not separated by the partition. Our requirement for $\net$
is not quite as strong. The partition is constructed randomly and
it will be sufficient to have each pair in $\net$ being separated by the partition with probability
at most $\eps$.


\subsection{Every Centre is Close to Some Pair in $\cents$}
The following says that if at least one centre of each pair in $\cents$ is open after a swap, then every centre
 in $\fopt \cup \floc$ is somewhat
close to some open centre. The bound is a bit big, but it will be multiplied by $O(\eps)$ 
whenever it is used in the local search analysis.
\begin{lemma}\label{lem:cents}
Let $A \subseteq \fopt \cup \floc$ be such that $A \cap \{\cent(i), i\} \neq \emptyset$
for each $(\cent(i), i) \in \cents$. Then $\del(i', A) \leq 5 \cdot D_{i'}$ for any $i' \in \opt \cup \loc$.
\end{lemma}
\begin{proof}
We first prove the statement for $i' \in \opt$, the other case is similar but requires one additional step so we will discuss it below.
Consider the following sequence of centres. Initially, set $i_0 := i'$, $i_1 := \eta(i_0)$, and $i_2 := \ophi(i_1)$.
We build the rest inductively, noting that we guarantee $i_a \in \ophi(\fopt)$ for even indices $a \geq 2$
(so $\cent(i_a)$ is defined).

Inductively, for even $a \geq 2$ we do the following. If $i_a \in A$ then we stop. Otherwise,
 if $(\cent(i_a),i_a) \in \cents$
then by assumption it must be that $\cent(i_a) \in A$ so we let $i_{a+1} := \cent(i_a)$ and stop.
Finally, if $(\cent(i_a), i_a) \not\in \cents$ then we set $i_{a+1} := \ophi(i_a)$ and $i_{a+2} := \ophi(i_{a+1})$
and iterate with $a' = a+2$.
This walk is depicted in Figure \ref{fig:longpath}.




We will soon show the walk terminates. For now, we observe that, apart from the first step, the steps decrease in length geometrically. In particular, consider some $a \geq 2$ 
such that the walk did not stop at $i_a$. If $i_{a+1} = \ophi(i_a)$ then
\[ \del(i_a, i_{a+1}) = \del(i_a, \ophi(i_a)) \leq \del(i_a, i_{a-1}). \]
If $i_{a+1} \neq \ophi(i_a)$ then it must be $i_a = \cent(i_{a-1})$. In this case,
it must be $i_a = \ophi(i_{a-1})$, so because $\cent(i_a)$ is the closest centre in $\ophi^{-1}(i_a)$ to $i_a$ (by definition), we have
\[ \del(i_a, i_{a+1}) = \del(i_a, \cent(i_a)) \leq \del(i_a, i_{a-1}). \]

We prove that the lengths of the edges traversed decrease geometrically with every other step. This, along with the fact that the lengths of the steps of the walk are nonincreasing (except, perhaps, the first two steps),
will show that this process eventually terminates and also bounds the cost of the path.

\begin{claim}\label{claim:geom}
For every even $a \geq 2$ such that the walk did not end at $i_a$ or $i_{a+1}$, we have $\del(i_a, i_{a+1}) \leq 2\eps \cdot \del(i_{a-1}, i_a)$.
\end{claim}
\begin{proof}
Because the walk did not end at $i_a$ or $i_{a+1}$, 
$(\cent(i_a), i_a) \not\in \cents$ meaning $D_{i_a} < \eps \cdot \del(\cent(i_a), i_a) \leq \eps \cdot \del(i_{a-1}, i_a)$.
%
Using this and Lemma 2 in the first bound below, we see
\[ \del(i_a, i_{a+1}) \leq (1+\eps) \cdot D_{i_a} < \eps(1 + \eps) \cdot \del(i_{a-1}, i_a) \leq 2\eps \cdot \del(i_{a-1}, i_a) \]
\end{proof}


\begin{figure}[t]
\begin{center}
\includegraphics[scale=1.4]{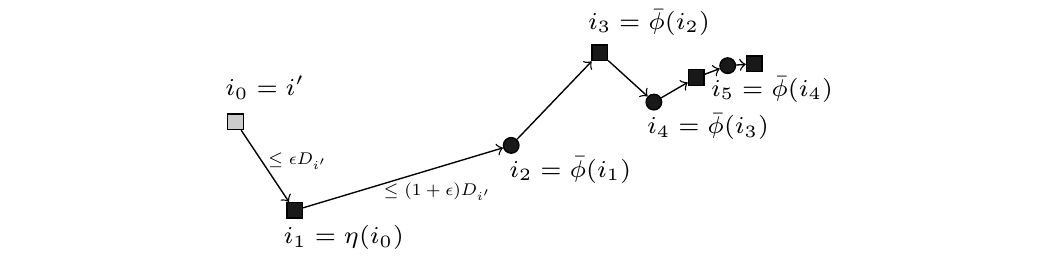} 
\end{center}
 \caption{Illustration of the walk $i_0, i_1, \ldots$ in the proof of Lemma 3. Apart from $a = 1$, the step $(i_a, i_{a+1})$ is shorter than the step $(i_{a-1}, i_a)$. Furthermore, every other step
 decreases in length geometrically.}
 \label{fig:longpath}
\end{figure}

Let $m$ be the index of the last centre in the walk.
Thus, $\del(i_a, i_{a+1}) \leq \del(i_{a-1}, i_a)$ for all $2 \leq a \leq m-1$. From this and Claim \ref{claim:geom} we have
\begin{eqnarray*}
\del(i', A) & \leq & \sum_{j=0}^{m-1} \del(i_j, i_{j+1}) \\
& \leq & \del(i_0, i_1) + \sum_{\substack{2 \leq a \leq m\\ a {\rm ~even}}} 2\del(i_{a-1}, i_a) + \del(i_{m-1}, i_m) \\
& \leq & \del(i_0, i_1) + 2\del(i_1, i_2) \sum_{a \geq 0} (2\eps)^a + \del(i_{m-1}, i_m) \\
& \leq & \del(i_0, i_1) + \frac{2}{1-2\eps} \cdot \del(i_1, i_2) + \del(i_1, i_2) \\
& = & \del(i_0, i_1) + \frac{3-2\eps}{1-2\eps}\cdot \del(i_1, i_2) \\
& \leq & \eps \cdot D_{i'} + \frac{(3-2\eps)(1 + \eps)}{1-2\eps} \cdot D_{\eta(i')} \\
& \leq & (3+7\eps) \cdot D_{i'}.
\end{eqnarray*}
The second last step uses Lemma \ref{lem:phi} and the
last step uses the fact that $D_{\eta(i')} \leq D_{i'}$ (either $i' = \eta(i')$ or else $i'$ was filtered out by $\eta(i')$, in which case it has a larger $D$-value)
and the assumption that $\eps$ is small enough.

Now suppose $i' \in \loc$. We bound $\del(i', A)$ mostly using what we have done already.
That is, we have
\[ \del(i', A) \leq \del(i', \phi(i')) + \del(\phi(i'), A) \leq D_{i'} + (3+7\eps)D_{\phi(i')}. \]
Note $D_{\phi(i')} \leq \del(i', \phi(i')) = D_{i'}$ again by Lemma 2. So,
\[ \del(i', A) \leq (4+7\eps) D_{i'} \leq 5D_{i'}. \]
\end{proof}


\subsection{Good Partitioning of $\opt \cup \loc$ and Proof of Theorem \ref{thm:euclid}}
The main tool used in our analysis is the existence of the following randomized partitioning scheme.
\begin{theorem}\label{thm:partition}
There is a randomized algorithm that samples a partitioning $\pi$ of $\opt \cup \loc$ such that:
\begin{itemize}
\item For each part $P \in \pi$, $|P \cap \opt| = |P \cap \loc| \leq \swaps$.
\item For each part $P \in \pi$, $\loc \triangle P$ includes at least one centre from every pair in $\cents$.
\item For each $(i^*, i) \in \net$, $\Pr[i, i^* {\rm ~lie~in~different~parts~of~}\pi] \leq \eps$.
\end{itemize}
\end{theorem}
We prove this theorem in Section \ref{sec:partitioning}. For now, we will complete the analysis of the 
local search algorithm using this partitioning scheme. Note that in the following we do not use the geometry of the
metric (i.e. all arguments hold for general metrics); it is only in the proof of Theorem \ref{thm:partition} that
we use properties of $\RR^d$.

The following gives a way to handle the fact that the triangle inequality does not hold with squares of the distances.
\begin{lemma}\label{lem:sqr}
For any real numbers $x,y$ we have $(x+y)^2 \leq 2(x^2 + y^2)$.
\end{lemma}
\begin{proof}
$(x+y)^2 \leq (x+y)^2 + (x-y)^2 = 2x^2 + 2y^2.$
\end{proof}

\begin{lemma} \label{lem:client_rad}
For each point $j \in \cX$, $D_{\osigma(j)} \leq D_{\sigma(j)} \leq \del(j, \sigma(j)) + \del(j, \sigma^*(j))$. Similarly,
$D_{\osigma^*(j)} \leq D_{\sigma^*(j)} \leq \del(j, \sigma(j)) + \del(j, \sigma^*(j))$.
\end{lemma}
\begin{proof}
As usual, we only prove the first statement since the second is nearly identical. If $\osigma(j) = \sigma(j)$ then $D_{\osigma(j)} = D_{\sigma(j)}$
is trivially true. Otherwise, $\osigma(j) = \eta(\sigma(j))$ was already in $\floc$ when $\sigma(j)$ was considered by the filtering algorithm meaning $D_{\osigma(j)} \leq D_{\sigma(j)}$.

For the other inequality, note
\[ D_{\sigma(j)} = \del(\sigma(j), \phi(\sigma(j))) \leq \del(\sigma(j), \sigma^*(j)) \leq \del(j, \sigma(j)) + \del(j, \sigma^*(j)). \]
\end{proof}

\begin{pproof}{Theorem \ref{thm:euclid}}\\
Let $\pi$ be a partition sampled by the algorithm from Theorem \ref{thm:partition}.
For each point $j \in \cX$ and each part $P$ of $\pi$, let $\Delta^P_j := \del(j, \loc \triangle P)^2 - \del(j, \loc)^2$ 
denote the change in assignment cost for the point
after swapping in the centers in $P\cap \opt$ and swapping out $P\cap \loc$. 
Local optimality of $\loc$ and $|P \cap \loc| = |P \cap \opt| \leq \rho(\eps, d)$ means $0 \leq \sum_j \Delta^P_j$ for any part $P$.

Classify each point $j \in \cX$ in one of the following ways:
\begin{itemize}
\item {\bf Lucky}: $\sigma(j)$ and $\osigma(j)$ do not lie in the same part of $\pi$.
\item {\bf Long}: $j$ is not lucky but $\del(\osigma(j), \osigma^*(j)) > \eps^{-1} \cdot D_{\osigma(j)}$.
\item {\bf Bad}: $j$ is not lucky or long and $(\osigma^*(j), \osigma(j)) \in \net$ yet $\osigma(j), \osigma^*(j)$ lie in different parts of $\pi$.
\item {\bf Good}: $j$ is neither lucky, long, nor bad.
\end{itemize}

We now place an upper bound on $\sum_{P \in \pi} \Delta^P_j$ for each point $j\in\cX$.
Note that each centre in $\loc$ is swapped out exactly once over all
swaps $P$ and each centre in $\opt$ is swapped in exactly once. With this in mind,
consider the following cases for a point $j \in \cX$.
In the coming arguments, we let $\del_j := \del(j, \sigma(j))$ and $\del^*_j := \del(j, \sigma^*(j))$ for brevity.
Note $\costl_j = \del_j^2$ and $\costo_j = \del^{*2}_j$.

In all cases for $j$ except when $j$ is bad, the main idea is that we can bound the distance from $j$ to some point in $\loc \triangle P$ by first moving it to either $\sigma(j)$ or $\sigma^*(j)$
and then moving it a distance of $O(\eps) \cdot (\del_j + \del^*_j)$ to reach an open facility. Considering that we reassigned $j$ from $\sigma(j)$, the reassignment cost will be
\[ (\del_j + O(\eps) \cdot (\del_j + \del^*_j))^2 - c_j = O(\eps) \cdot (c_j + c^*_j)\]
or
\[ (\del^*_j + O(\eps) \cdot (\del_j + \del^*_j))^2 - c_j = (1 + O(\eps)) \cdot c^*_j - (1 - O(\eps)) \cdot c_j. \]

~

\noindent
{\bf Case: $j$ is lucky}\\
For the part $P \in \pi$ with $\sigma^*(j) \in P$, we have $\Delta^P_j \leq c^*_j - c_j$ as we could move $j$ from $\sigma(j)$ to $\sigma^*(j)$. If $\sigma(j)$ is swapped out in a different swap $P'$,
we move $j$ to $\osigma(j)$ (which remains open because $j$ is lucky) and bound $\Delta^{P'}_j$ by:
\begin{eqnarray*}
\Delta^{P'}_j & \leq & \del(j, \osigma(j))^2 - \del_j^2 \\
& \leq & (\del_j + \del(\sigma(j), \osigma(j)))^2 - \costl_j \\
& \leq & (\del_j + \eps \cdot D_{\sigma(j)})^2 - \costl_j {\hspace{27.5mm}}\mbox{(Lemma \ref{lem:filter})} \\
& = & 2\eps \cdot \del_j \cdot D_{\sigma(j)} + \eps^2 \cdot D_{\sigma(j)}^2 \\
& \leq & 2\eps \cdot \del_j \cdot (\del_j + \del^*_j) + \eps^2 \cdot (\del_j + \del^*_j)^2 
\quad\quad\mbox{(Lemma \ref{lem:client_rad})}\\
& \leq & 2\eps\cdot (\del_j + \del^*_j)^2 + \eps^2\cdot (\del_j+\del^*_j)^2\\
&\leq & 4\eps \cdot (c^*_j + c_j) + 2\eps^2 \cdot (c^*_j + c_j) {\hspace{14mm}}\mbox{(Lemma \ref{lem:sqr})}\\
& \leq & 6 \eps \cdot (\costo_j + \costl_j),
\end{eqnarray*}
again using the assumption that $\eps$ is sufficiently small.
For every other swap $P''$, we have that $\sigma(j)$ remains open after the swap so $\Delta^{P''}_j \leq 0$ as we could just leave $j$ at $\sigma(j$).
In total, we have
\[ \sum_{P \in \pi} \Delta_j^P \leq \costo_j - \costl_j + 6\eps \cdot (\costo_j + \costl_j).\]

~

\noindent
{\bf Case: $j$ is long}\\
Again, for $P \in \pi$ with $\sigma^*(j)\in P$
we get $\Delta^P_j \leq \costo_j - \costl_j$. If $\sigma(j)$ is swapped out in a different swap $P'$, then
we bound $\Delta^{P'}_j$ by moving $j$ from $\sigma(j)$ to the open centre nearest to $\sigma(j)$.
Note that $\loc \triangle P'$ contains at least one centre from every pair in $\cents$, so we bound this distance
using Lemma \ref{lem:cents}. This case is depicted in Figure \ref{fig:caselong}.
We have
\begin{eqnarray*}
 D_{\sigma(j)} &\leq &\del(\sigma(j), \osigma(j)) + \del(\osigma(j), \phi(\osigma(j)))\\
 &\leq & \eps D_{\sigma(j)} + D_{\osigma(j)} 
                     \hspace{6.4cm}\mbox{(Lemma \ref{lem:filter})}\\
 &\leq & \eps(\del(j,\sigma(j))+\del(j,\sigma^*(j)) + \eps\del(\osigma(j),\osigma^*(j))
                  \quad \quad\quad\quad\quad\mbox{(since $j$ is long)}\\
 &\leq & \eps(\del_j+\del^*_j) + \eps(\del(\osigma(j),\sigma(j))+\del(\sigma(j),j)+\del(j,\sigma^*(j)) + 
                   \del(\sigma^*(j),\osigma^*(j)) )\\ 
 &\leq& \eps(\del_j+\del^*_j) + \eps(\eps D_{\sigma(j)}+\del_j + \del^*_j + \eps D_{\sigma^*(j)})
                                                         \quad\quad \quad\quad\quad\mbox{(Lemma \ref{lem:filter})}\\
 &\leq& \eps(\del_j+\del^*_j) + \eps(\eps(\del_j+\del^*_j) + \del_j + \del^*_j + \eps(\del_j+\del^*_j))
                                                          \hspace{9.5mm}\mbox{(Lemma \ref{lem:client_rad})}\\
 & = & 2\eps(1+\eps)(\del_j+\del^*_j).
\end{eqnarray*}
Using this, we bound $\Delta^{P'}_j$ as follows.
\begin{eqnarray*}
\Delta^{P'}_j & \leq & (\del_j + \del(\sigma(j), \loc \triangle P))^2 - \costl_j \\
& \leq & (\del_j + 5D_{\sigma(j)})^2 - \costl_j  \quad\quad\quad\quad\quad\quad\quad\quad\quad\quad\mbox{(Lemma \ref{lem:cents})} \\
& \leq & (\del_j + 10\eps(1+\eps) \cdot (\del_j + \del^*_j))^2 - \costl_j \quad\quad\quad\quad \mbox{bound for $D_{\sigma(j)}$ above}\\
& \leq & 21 \eps \cdot \del_j (\del_j + \del^*_j) + 101 \eps^2 \cdot (\del_j + \del^*_j)^2 \\
& \leq & 22\eps \cdot (\del_j + \del^*_j)^2 \\
& \leq & 44\eps \cdot (\costl_j + \costo_j).
\end{eqnarray*}

\begin{figure}[h!]
\begin{center}
\includegraphics[scale=1]{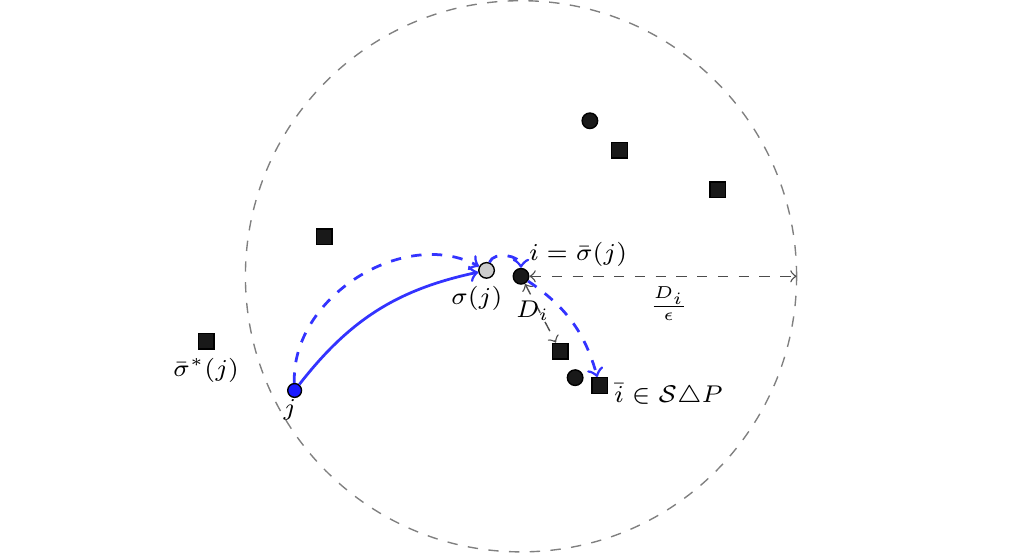} 
\end{center}
 \caption{Illustrating the case when $j$ is long. It is moved from $\sigma(j)$ to the nearest open centre at an additional distance of $O(D_{\sigma(j)})$.
 This is negligible compared to $\del_j + \del^*_j$.}
 \label{fig:caselong}
\end{figure}

In every other swap we could leave $j$ at $\sigma(j)$.
Thus, for a long point $j$ we have
\[ \sum_{P \in \pi} \Delta_j^P \leq \costo_j - \costl_j + 44\eps \cdot (\costo_j + \costl_j).\]

%
%
~

\noindent
{\bf Case: $j$ is bad}\\
We only move $j$ in the swap $P$ when $\sigma(j)$ is closed. In this case,
we assign $j$ in the same way as if it was long. Our bound is weaker here and introduces significant positive dependence on $\costl_j$. This will eventually be
compensated by the fact
that $j$ is bad with probability at most $\eps$ over the random choice of $\pi$. For now, we just provide the reassignment cost bound for bad $j$.

\begin{eqnarray*}
\Delta^P_j & \leq & (\del_j + \del(\sigma(j), \loc \triangle P))^2 - \costl_j \\
& \leq & (\del_j + 5D_{\sigma(j)})^2 - c_j \quad\quad\quad\quad\quad\quad\quad\mbox{(Lemma \ref{lem:cents})} \\
& \leq & (6\del_j + 5\del^*_j)^2 - c_j{\hspace{30.25mm}}\mbox{(Lemma \ref{lem:sqr})} \\
& \leq & 71 \cdot (\costo_j + \costl_j).
\end{eqnarray*}
So for bad points we have
\[ \sum_{P \in \pi} \Delta_j^P  \leq 71 \cdot (\costo_j + \costl_j). \]

~

\noindent
{\bf Case: $j$ is good}\\
This breaks into two subcases. We know $\del(\osigma^*(j), \osigma(j)) \leq \eps^{-1} \cdot D_{\osigma(j)}$ because $j$ is not long.
In one subcase, $D_{\osigma^*(j)} \geq \eps \cdot D_{\osigma(j)}$ so $(\osigma^*(j), \osigma(j)) \in \net$. Since $j$ is not bad and not lucky,
we have $\sigma(j), \osigma(j), \osigma^*(j) \in P$ for some common part $P \in \pi$. In the other subcase, $D_{\osigma^*(j)} < \eps \cdot D_{\osigma(j)}$.
Note in this case we still have $\sigma(j), \osigma(j) \in P$ for some common part $P$ because $j$ is not lucky.

~

\noindent
{\bf Subcase: $D_{\osigma^*(j)} \geq \eps \cdot D_{\osigma(j)}$}\\
The only time we move $j$ is when $\sigma(j)$ is closed. As observed in the previous paragraph, this happens in the same swap when $\osigma^*(j)$ is opened,
so send $j$ to $\osigma^*(j)$. This is illustrated in Figure \ref{fig:casegood2}.

\begin{figure}[h!]
\begin{center}
\includegraphics[scale=1]{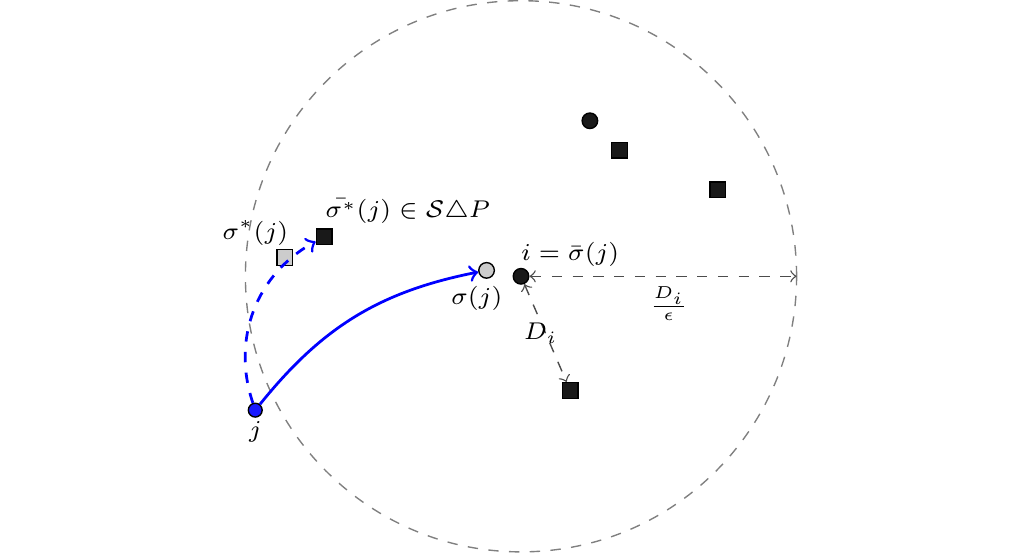} 
\end{center}
 \caption{Illustrating the subcase when $D_{\osigma^*(j))} \geq \eps \cdot D_{\osigma(j)}$. In this case, the only additional distance $j$ travels after first being
 reassigned to $\sigma^*(j)$ is to the nearby proxy $\osigma^*(j)$}
 \label{fig:casegood2}
\end{figure}

\begin{eqnarray*}
\Delta^P_j & \leq & (\del^*_j + \del(\sigma^*(j), \osigma^*(j)))^2 - \costl_j \\
& \leq & (\del^*_j + \eps \cdot D_{\sigma^*(j)})^2 - \costl_j \hspace{4.5cm}\mbox{(Lemma \ref{lem:filter})}\\
& \leq & (\del^*_j + \eps \cdot(\del_j + \del^*_j))^2 - \costl_j \\
& = & \costo_j + \eps^2\cdot(\del_j+\del^*_j)^2 + 2\eps \cdot \del^*_j(\del_j+\del^*_j) - \costl_j \\
& \leq & \costo_j + 2\eps^2\cdot(\costl_j+\costo_j) + 2\eps \cdot(\del_j+\del^*_j)^2-\costl_j\quad\quad\quad\quad\mbox{(Lemma \ref{lem:sqr})}\\
& \leq & \costo_j - \costl_j + 5\eps \cdot (\costo_j + \costl_j). \hspace{4.3cm}\mbox{(Lemma \ref{lem:sqr})}
\end{eqnarray*}

~

\noindent
{\bf Subcase: $D_{\osigma^*(j)} < \eps \cdot D_{\osigma(j)}$}\\
Again, the only time we move $j$ is when $\sigma(j)$ is closed.
We reassign $j$ by first moving it to $\sigma^*(j)$ and then using Lemma \ref{lem:cents} to further bound the cost.
Figure \ref{fig:casegood1} depicts this reassignment.

\begin{figure}[h!]
\begin{center}
\includegraphics[scale=1]{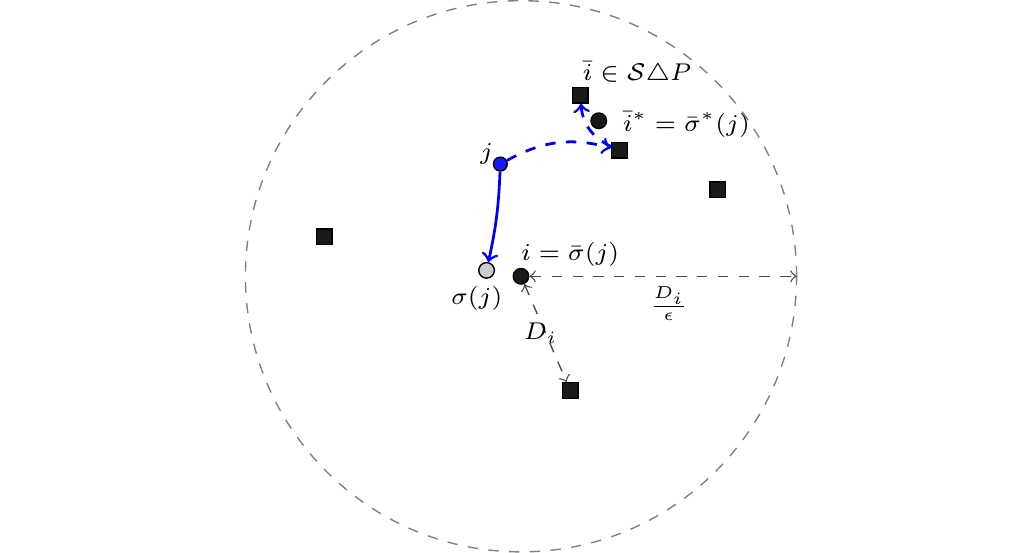} 
\end{center}
 \caption{Illustrating the subcase when $D_{\osigma^*(j)} < \eps \cdot D_{\osigma(j)}$. The only additional distance $j$ moves after being reassigned to $\sigma^*(j)$ is
small, as it travels first to the nearby proxy $\osigma^*(j)$ and then at distance at most $5D_{\osigma^*(j)} \leq 5\eps \cdot D_{\osigma(j)}$ by Lemma \ref{lem:cents}.}
 \label{fig:casegood1}
\end{figure}

We bound the cost change for this reassignment as follows. Recall that 
$\delta(\sigma^*(j),\osigma^*(j))\leq \eps D_{\sigma^*(j)}$ by our filtering. This implies that
 $D_{\sigma^*(j)} \leq D_{\osigma^*(j)} + \eps D_{\sigma^*(j)}$ which in turn implies
$D_{\sigma^*(j)} \leq \frac{1}{1-\eps}D_{\osigma^*(j)}$; thus 
$D_{\sigma^*(j)} \leq (1 + 2\eps) D_{\osigma^*(j)}$ 
which in turn is bounded by $\eps(1 + \eps) D_{\osigma(j)}$ in this subcase (by the assumption of subcase). Thus,
\begin{eqnarray*}
\Delta^P_j & \leq & (\del^*_j + \del(\sigma^*(j), \loc \triangle P))^2 - \costl_j \\
& \leq & (\del^*_j + 5 D_{\sigma^*(j)}) - \costl_j \\
& \leq & (\del^*_j + 5\eps(1+2\eps) \cdot D_{\osigma(j)})^2 - \costl_j \\
& \leq & (\del^*_j + 5\eps (1+2\eps) \cdot (\del_j + \del^*_j))^2 - \costl_j 
                                      \quad\quad\quad\quad\mbox{(Lemma \ref{lem:client_rad})}\\
& \leq & \costo_j + 11\eps \cdot \del^*_j (\del_j + \del^*_j) + 26\eps^2 (\del_j + \del^*_j)^2 - \costl_j \\
& \leq & \costo_j - \costl_j + 12\eps \cdot (\del_j + \del^*_j)^2 \\
& \leq & \costo_j - \costl_j + 24\eps \cdot (\costo_j + \costl_j).
\end{eqnarray*}


Considering both subcases we can say that for any good point $j$ that
\[ \sum_{P \in \pi} \Delta_j^P \leq \costo_j - \costl_j + 24\eps \cdot (\costo_j + \costl_j). \]

Aggregating these bounds and remembering $0 \leq \sum_{j \in \cX} \Delta^P_j$ for each $P \in \pi$ because $\loc$ is a locally optimum solution, 
we have
\[ 0 \leq \sum_{j \in \cX} \sum_{P \in \pi} \Delta_j^P \leq \sum_{\substack{j \in \cX \\ j {\rm ~not~bad}}} 
\left[(1 + 44\eps) \costo_j - (1 - 44\eps) \costl_j\right] + \sum_{\substack{j \in \cX \\ j {\rm ~bad}}} 71 (\costo_j + \costl_j). \]

The last step is to average this inequality over the random choice of $\pi$. Note that any point $j \in \cX$ 
is bad with probability at most $\eps$ by the guarantee in Theorem \ref{thm:partition} and the definition of {\em bad}.
Thus, we see
\begin{eqnarray*}
0 & \leq & \mathbf{E}_{\pi}\left[\sum_{j \in \cX} \sum_{P \in \pi} \Delta_j^P\right] \\
& \leq & \sum_{j \in \cX} \Pr[j {\rm ~not~bad}] \cdot \left[(1 + 44\eps) \costo_j - (1 - 44\eps) \costl_j\right] + 
                 \Pr[j {\rm ~bad}] \cdot 71 (\costo_j + \costl_j) \\
& \leq & \sum_{j \in \cX} (1 + 44\eps) \costo_j - (1-\eps) \cdot (1-44\eps) \costl_j + 71\eps \cdot (\costo_j + \costl_j) \\
& \leq & \sum_{j \in \cX} (1 + 115\eps) \costo_j - (1 - 116\eps) \costl_j. \\
\end{eqnarray*}
Rearranging shows
\[ \cost(\loc) \leq \frac{1 + 115\eps}{1 - 116\eps} \cdot \cost(\opt) \leq (1 + O(\eps)) \cdot \cost(\opt). \]

\end{pproof}

\subsection{Running Time Analysis}
Recall that we can go to the discrete case by finding a set $\fa$ of size $O(n\epsilon^{-d}\log(1/\epsilon)$.
The analysis of Arya et al.~\cite{AGKMMP01,AGKMMP04} shows that the number of local search
steps is at most $\log(cost(S_0)/cost(\opt))/\log\frac{1}{1-\eps/k}$, where $S_0$ is the initial solution.
This is polynomial in the total bit complexity of the input (i.e. the input size). So we focus on bounding the time complexity of each
local search step. Since the number of swaps in each step is bounded by $\rho=\swaps$, a crude
upper bound on the time complexity of each step is $O((n\eps^{-d}\log(1/\eps))^\rho)$.

We can speed up this algorithm by using the idea of coresets. 
First observe that our local search algorithm extends to the weighted setting where each point $j\in\cX$ has
a weight $w(j)$ and the cost of a clustering with centres $C$ is $\sum_{j\in\cX} w(j)\cdot\del^2(j,C)$.

Let $\cX\subseteq\RR^d$ be a weighted set of points with weight function $w:\cX\rightarrow\RR$.
For any set $C\subseteq \RR^d$ of size $k$ we use $\mu_C(\cX)$ to denote $\sum_{j\in\cX} w(x)\del(j,C)^2$.
The cost of the optimum \kmeans solution for $\cX$ is $\mu_{\opt}(\cX,k)=\min_{C\subseteq\RR^d:|C|=k} \mu_C(\cX)$.


\begin{definition}
Let $\cX\subseteq\RR^d$ and $S\subseteq\RR^d$ be two weighted point sets.
For any $k,\eps$, we say $S$
is a $(k,\eps)$-coreset if for every set of $k$ centres $C\subseteq \RR^d$:
\[(1-\eps)\mu_C(\cX)\leq\mu_C(S)\leq (1+\eps)\mu_C(\cX).\]
\end{definition}

A set $Z\subseteq \RR^d$ is
a $(k,\eps)$-centroid set for $\cX$ if there is a set $C\subseteq Z$ of size $k$ such that
$\mu_C(\cX)\leq (1+\eps)\mu_{\opt}(\cX,k)$.

Earlier works on coresets for $k$-means \cite{HK05,FMS07,FSS13,LS10} imply the existence of a $(k,\eps)$-coresets
of small size. For example,~\cite{HK05} show existence of $(k,\eps)$-coresets of size $O(k^3/\eps^{d+1})$.
Applying the result of~\cite{Matousek00}, we get a $(k,\eps)$-centroid set of size $O(\log(1/\eps)k^3/\eps^{2d+1})$
over a $(k,\eps)$-coreset of size $O(k^3/\eps^{d+1})$.
Thus running our local search $\rho$-swap algorithm takes 
$O((k/\eps)^\zeta)$ time per iteration where $\zeta = d^{O(d)}\cdot \eps^{-O(d/\eps)}$.



\section{The Partitioning Scheme: Proof of Theorem \ref{thm:partition}}\label{sec:partitioning}
Here is the overview of our partitioning scheme. First we bucket the real values by defining bucket $a$
to have real values $r$ where $\eps^{-a}\leq r< \eps^{-(a+1)}$. Fix $b=\Theta(1/\eps)$ and
then choose a random off-set $a'$ and consider $b$ consecutive buckets 
$a'+\ell\cdot b,\ldots,a'+(\ell+1)\cdot b - 1$ to form ``bands''.
We get a random partition of points $i\in\fopt\cup\floc$ into bands 
$B_1,B_2,\ldots$ such that all centres $i$ in the same band $B_\ell$
have the property that their $D_i$ values are in buckets $a'+\ell\cdot b,\ldots,a'+(\ell+1)\cdot b - 1$
(so  $\eps^{-(a'+\ell\cdot b)}\leq D_i <\eps^{-(a'+(\ell+1)\cdot b-1)}$).
Given that for each pair $(i,i^*)\in\net$ the $D_i,D_{i^*}$ values are within a factor $1/\eps$ of each other, 
the probability that they fall into different bands is small (like $\eps/4$).
Next we impose a random grid partition on each band $B_\ell$ to cut it into cells with cell width roughly 
$\Theta(d/\eps^{\ell\cdot b})$; again the randomness comes from choosing a random off-set for our grid. 
The randomness helps to bound the probability of any pair $(i,i^*)\in\net$ being cut into two different cells
to be small again. This will ensure the 3rd property of the theorem holds. Furthermore,
since each $i\in B_\ell$ has $D_i\geq \eps^{-(a'+\ell\cdot b)}$ and by the filtering we did
(Lemma \ref{lem:filter}), balls of radius $\frac{\eps}{2}\cdot\eps^{-(a'+\ell\cdot b)}$ around them must be disjoint; 
hence using a simple volume/packing argument we can bound the number of centres from each $B_\ell$ that fall into the 
same grid cell by a function of $\eps$ and $d$. This helps of establish the first property.
We have to do some clean-up to ensure that for each pair $(i,\cent(i))\in\cents$ they belong to the same part
(this will imply property 2 of the Theorem)
and that at the end for each part $P$, $|P\cap\loc|=|P\cap\opt|$. These are a bit technical and are described in details below.

Start by geometrically grouping the centres in $\fopt \cup \floc$. For $a \in \mathbb Z$, let
\[ G_a := \left\{i \in \fopt \cup \floc : \frac{1}{\eps^{a}} \leq D_i < \frac{1}{\eps^{a+1}}\right\}. \]
Finally, let $G_{-\infty} := \{ i \in \fopt \cup \floc : D_i = 0\}$. Note that each $i \in \fopt \cup \floc$ appears in exactly one set among $\{G_a : a \in \mathbb Z\} \cup \{G_{-\infty}\}$.

We treat $G_{-\infty}$ differently in our partitioning algorithm. It is important to note that no pair in $\cents$ or $\net$ has precisely one point in $G_{-\infty}$,
as the following shows.
\begin{lemma}\label{lem:zero}
For each pair of centres $(i^*, i) \in \cents \cup \net$, $|\{i, i^*\} \cap G_{-\infty}| \neq 1$.
\end{lemma}
\begin{proof}
Consider some colocated pair $(i^*, i) \in \loc \times \opt$. As $D_i = D_{i^*} = 0$ and because no other centre in $\loc \cup \opt$ is colocated with $i$ and $i^*$,
then $i \in \floc$ and $i^* \in \fopt$. Thus, $\ophi(i^*) = i$ and it is the unique closest facility in $\fopt$ to $i$ so $i^* = \cent(i)$. This shows every
pair $(i^*, i) \in \cents$ with either $D_i = 0$ or $D_{i^*} = 0$ must have both $D_i = D_{i^*} = 0$ so $i^*, i \in G_{-\infty}$.


%

Next consider some $(i^*, i) \in \net$. We know $\eps \cdot D_i \leq D_{i^*}$ by definition of $\net$. Thus, if $i^* \in G_{-\infty}$ then $i \in G_{-\infty}$ as well.
Conversely, suppose $i \in G_{-\infty}$. Since $(i^*, i) \in \net$ then $\del(i, i^*) \leq \eps^{-1} \cdot D_i = 0$.
So, $D_{i^*} = 0$ meaning $i^* \in G_{-\infty}$ as well.
\end{proof}

\subsection{Partitioning $\{G_a : a \in \mathbb Z\}$}


\noindent
{\bf Step 1: Forming Bands}\\
Fix $b$ to be the smallest integer that is at least $4/\eps$. We first partition $\{G_a : a \in \mathbb Z\}$ into {\em bands}.
Sample an integer {\em shift} $a'$ uniformly at random in the range $\{0, 1, \ldots, b-1\}$. For each $\ell \in \mathbb Z$, form the band $B_\ell$ as follows:
\[ B_{\ell} := \bigcup_{0 \leq j \leq b-1} G_{a' + j + \ell \cdot b}. \]

~

\noindent
{\bf Step 2: Cutting Out Cells}\\
Focus on a band $B_{\ell}$. Let $W_{\ell} := 4d \cdot \eps^{-(\ell+2) \cdot b - 1}$ 
(recall that $d$ is the dimension of the Euclidean metric). This will be the ``width'' of the cells we create.
It is worth mentioning that this is the first time we are going to use the properties of Euclidean metrics as 
all our arguments so far were treating $\delta(\cdot,\cdot)$ as a general metric.

We consider a random grid with cells having width $W_{\ell}$ in each dimension. Choose an {\em offset}
$\beta \in \mathbb R^d$ uniformly at random from the cube $[0, W_\ell]^d$. 
For emphasis, we let ${\bf p}^i_j$ refer to the $j$'th component of $i$'s point in Euclidean space
(this is the only time we will refer specifically to the coordinates for a point). So centre $i$
has Euclidean coordinates $({\bf p}^i_1, {\bf p}^i_2, \ldots, {\bf p}^i_d)$.

For any tuple of integers ${\bf a} \in \mathbb Z^d$ define the {\em cell} $C^\ell_{\bf a}$ as follows.
\[ C^\ell_{\bf a} = \{i \in B_{\ell} : \beta_j + W_{\ell} \cdot {\bf a}_j \leq {\bf p}^i_j < 
											\beta_j + W_{\ell} \cdot ({\bf a}_j + 1) {\rm ~for~all~} 1 \leq j \leq d\}. \]
These are all points in the band $B_{\ell}$ that lie in the half-open cube with side length $W_{\ell}$ and lowest corner
 $\beta + W_{\ell} \cdot {\bf a}$.
Each $i \in B_{\ell}$ lies in precisely one cell as these half-open cubes tile $\mathbb R^d$.

~

\noindent
{\bf Step 3: Fixing $\cents$}\\
Let $\mathcal I \subseteq \fopt$ be the centres of the form $\cent(i)$ for some $(\cent(i), i) \in \cents$ where $i, \cent(i) \not\in G_{-\infty}$ and $i$ and $\cent(i)$ lie in different cells
after step 2.
We will simply move each $i \in \mathcal I$ to the cell containing $\ophi(i)$.
More precisely, for $\ell \in \mathbb Z$ and ${\bf a} \in \mathbb Z$ we define the {\em part} $P^\ell_{\bf a}$ as
\[ P^\ell_{\bf a} = (C^{\ell}_{\bf a} - \mathcal I) \cup \{i \in \mathcal I : \ophi(i) \in C^{\ell}_{\bf a}\}. \]

We will show that the parts $P^\ell_{\bf a}$ essentially satisfy most of the desired properties stated about the random partitioning scheme promised in Theorem \ref{thm:partition}.
It is easy to see that property 2 holds for each part $P^\ell_{\bf a}$ since for each pair $(\cent(i),i)\in\cents$
both of them belong to the same part.
All that remains is to ensure they are {\em balanced} (i.e. have include the same number of centres in $\opt \cup \loc$) and to incorporate $G_{-\infty}$ and the centres
in $\opt \cup \loc$ that were filtered out. These are relatively easy cleanup steps.


\subsection{Properties of the Partitioning Scheme}
We will show the parts formed in the partitioning scheme so far have constant size (depending only on $\eps$ and $d$) and also that
each pair in $\net$ is cut with low probability.
Figure \ref{fig:grid} illustrates some key concepts in these proofs.

\begin{figure}[t]
\begin{center}
\includegraphics[scale=1.2]{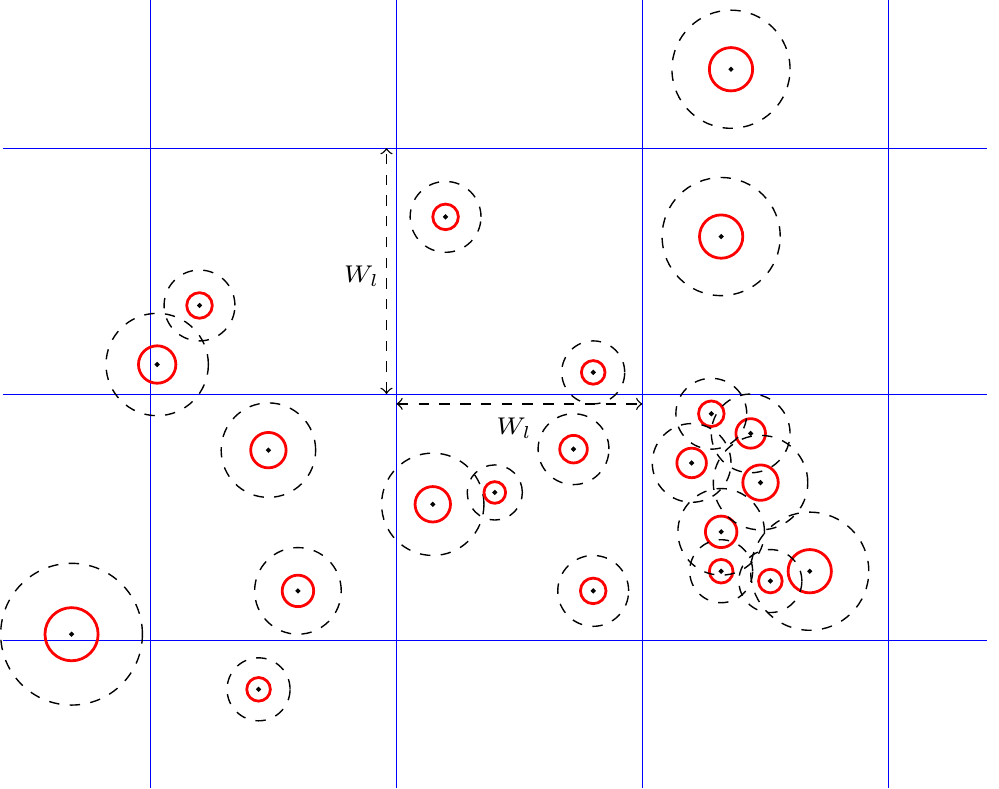} 
\end{center}
 \caption{The random grid \& balls of radii $\frac{\epsilon \cdot D_{i}}{2}$ (solid) and $\frac{D_{i}}{\epsilon}$ (dashed) around each 
$i \in B_{\ell}$.
 Note the balls of radius $\frac{\eps \cdot D_i}{2}$ are disjoint, which is used to bound $|C^\ell_{\bf a}|$ by a volume argument in Lemma \ref{lem:size}.
 For each pair $(i, i^*) \in \cents \cup \net$ in the band $B_\ell$, one centre must lie in the ball of radius $\frac{D_{i}}{\epsilon}$ around the other.
 The proof of Lemma \ref{lem:prob} essentially shows that each such ball crosses a grid line with very small probability, so such a pair is probably not cut by the random grid. }
 \label{fig:grid}
\end{figure}

\begin{lemma}\label{lem:size}
For each $\ell \in \mathbb Z$ and ${\bf a} \in \mathbb Z^d$ we have $|P^\ell_{\bf a}| \leq 2\cdot (2d)^{2d} \cdot \eps^{-9d/\eps}$.
\end{lemma}
\begin{proof}
Note that each centre $i \in C^\ell_{\bf a}$ witnesses the inclusion of at most one additional centre of $\mathcal I$
into $P^\ell_{\bf a}$ (in particular, $\cent(i)$).
So, $|P^\ell_{\bf a}| \leq 2 \cdot |C^\ell_{\bf a}|$ meaning
it suffices to show $|C^\ell_{\bf a}| \leq (2d)^{2d} \cdot \eps^{-9d/\eps}$.
For each $i \in C^\ell_{\bf a}$ we have $D_i \geq \eps^{-\ell \cdot b}$. By Lemma \ref{lem:filter}, the balls of radius $\frac{\eps}{2} \cdot \eps^{-\ell \cdot b}$ around the centres
in $C^\ell_{\bf a}$ must be disjoint subsets of $\mathbb R^d$.

The volume of a ball of radius $R$ in $\mathbb R^d$ can be lower bounded by $\left(\frac{\sqrt{2\pi}}{d}\right)^d \cdot R^d$, so
the total volume of all of these balls of radius $\frac{\eps}{2} \cdot \eps^{-\ell \cdot b}$ is at least
\[ |C^\ell_{\bf a}| \cdot \left(\frac{\sqrt{\pi} \cdot \eps}{\sqrt 2 \cdot d}\right)^d \cdot \eps^{-d \cdot \ell \cdot b}. \]

On the other hand, these 
balls have their centres
located in a common cube with side length $W_\ell$ and, thus, 
all such balls are contained in a slightly larger cube with side length $W_\ell + D_i \leq W_\ell + \eps^{-(\ell+2)\cdot b + 1}$.
The total volume of all balls is at most
$(W_\ell+ \epsilon^{-(\ell+2)\cdot b + 1})^d \leq (5d\eps^{-(2b+1)})^d \cdot \epsilon^{-d \cdot \ell \cdot b}$.
Combining this with the lower bound on the total volume, we see
\[ |C^\ell_{\bf a}| \leq \left(\frac{\sqrt{50} \cdot d^2 \cdot \eps^{-(2b+2)}}{\sqrt{\pi}} \right)^d \leq (2d)^{2d} \cdot \eps^{-9d/\eps}. \]
The last bound follows for sufficiently small $\epsilon$ and because $b$ is the smallest integer at least $4/\eps$.
\end{proof}

\begin{lemma}\label{lem:prob}
For each $(i^*, i) \in \net$ with  $i,i^* \not\in G_{-\infty}$ we have $\Pr[i, i^* {\rm ~lie~in~different~parts}] \leq \eps$.
\end{lemma}
\begin{proof}
We first bound the probability they lie in different bands. Note 
that $D_{i^*} \leq \del(i, i^*) \leq D_i / \eps$ and also $\eps D_i \leq D_{i^*}$ because
$(i^*, i) \in \net$. Thus, if we say $i \in G_a$ and $i^* \in G_{{a^*}}$ then $|a-{a^*}| \leq 1$. The probability that $G_a$ and $G_{{a^*}}$ are separated
when forming the bands $B_\ell$ is at most $1/b \leq \eps/4$.

Next, conditioned on the event that $i, i^*$ lie in the same band $B_\ell$, we bound the probability they lie in different cells. 
Say $i, i^* \in B_\ell$
and note
\[ \del(i, i^*) \leq \eps^{-1} \cdot D_i \leq \eps^{-(\ell+2) \cdot b} \leq \frac{\eps}{4d} \cdot W_{\ell}.\]
The probability that $i$ and $i^*$ are cut by the random offset
along one of the $d$-dimensions is then at most $\frac{\eps}{4d}$. Taking the union bound over all dimensions, the probability that $i$ and $i^*$ lie
in different cells is at most $\frac{\eps}{4}$.

Finally, we bound the probability that $i^*$ will be moved to a different part when fixing $\cents$. If $(i^*, \ophi(i^*)) \not\in \cents$ or if $i^* \not\in \mathcal I$
then this will not happen.
So, we will bound $\Pr[i^* \in \mathcal I]$ if $(i^*, \ophi(i^*)) \in \cents$.
That is, we bound the probability that $i', i^*$ lie in different parts where $i'$ is such that $i^* = \cent(i')$ and $(i^*, i') \in \cents$.
Note that $D_{i'} \leq \del(i', i^*) \leq (1+\eps)\cdot D_{i^*} \leq \eps^{-1} \cdot D_{i^*}$ by Lemma \ref{lem:phi} and
$D_{i'} \geq \eps \cdot \del(i', i^*) = \eps \cdot D_{i^*}$ by definition of $\cents$. So $i'$ and $i^*$ lie in different bands with probability
at most $\frac{\eps}{4}$ by the same argument as with $(i, i^*)$. Similarly, conditioned on $i', i^*$ lying in the same band, the probability they lie in different cells is at
most $\frac{\eps}{4}$ again by the same arguments as with $(i, i^*)$.

Note that if $i, i^*$ lie in different parts then they were cut by the random band or by the random box, or else $i', i^*$ were cut by the random band
or the random box (the latter may not apply if $i^*$ is not involved in a pair in $\cents$). Thus, by the union bound we have
\begin{eqnarray*}
\Pr[i, i^* {\rm ~are~in~different~parts}] & \leq & \Pr[i, i^* {\rm ~lie~in~different~} B_\ell] + \\
 && \Pr[i, i^* {\rm ~lie~in~different~} C^\ell_{\bf a} ~|~ i, i^* {\rm ~lie~in~the~same~} B_\ell] + \\
& & \Pr[i^*, i' {\rm ~lie~in~different~} B_\ell] + \\
&&\Pr[i^*, i' {\rm ~lie~in~different~} C^\ell_{\bf a} ~|~ i^*, i' {\rm ~lie~in~the~same~} B_\ell ] \\
& \leq & \frac{\eps}{4} + \frac{\eps}{4} + \frac{\eps}{4} + \frac{\eps}{4} = \eps 
\end{eqnarray*}
\end{proof}


\subsection{Balancing the Parts}\label{sec:balancing}
We have partitioned $\{G_a : a \in \mathbb Z\}$ into parts $P^\ell_{\bf a}$ for various $\ell \in \mathbb Z$ and ${\bf a} \in \mathbb Z^d$. We extend this to a partition of all of $\opt \cup \loc$
into constant-size parts that are balanced between $\opt$ and $\loc$ to complete the proof of Theorem \ref{thm:partition}.
Let $\mathcal P = \{P^\ell_{\bf a} : \ell \in \mathbb Z, {\bf a} \in \mathbb Z^d {\rm ~and~} P^\ell_{\bf a} \neq \emptyset\}$ be the collection of nonempty parts formed so far.

The proof of Lemma \ref{lem:zero} shows that $G_{-\infty}$ partitions naturally into colocated centres. So let $\mathcal P_{-\infty}$ denote the partition of $G_{-\infty}$
into these pairs.
Finally let $\mathcal P'$ denote the partition of $(\opt - \fopt) \cup (\loc - \floc)$ into singleton sets.

We summarize important properties of $\mathcal P \cup \mathcal P_{-\infty} \cup \mathcal P'$.
\begin{itemize}
\item $\mathcal P \cup \mathcal P_{-\infty} \cup \mathcal P'$ is itself a partitioning of $\opt \cup \loc$ into parts with size at most
$2\cdot (2d)^{2d} \cdot \eps^{-9d/\eps}$ (Lemma \ref{lem:size}).
\item Every $(\cent(i), i) \in \cents$ pair has both endpoints
in the same part. This is true for pairs with endpoints in $\mathcal P$ by Step 3 in the partitioning of $\{G_a : a \in \mathbb Z\}$.
If $i \in G_{-\infty}$ this is trivially true as colocated pairs $(i^*, i) \in \fopt \cup \floc$ must have $i^* = \cent(i)$.
\item Over the randomized formation of $\mathcal P$, each $(i^*, i) \in \net$ has endpoints in different parts with probability at most $\eps$.
For $i, i^*$ lying in $\mathcal P$ this follows from Lemma \ref{lem:prob}. For $i, i^*$ lying in $\mathcal P_{-\infty}$ this follows because
they must then be colocated pairs so they always form a part by themselves.
\end{itemize}

From now on, we simply let $\overline{\mathcal P} = \mathcal P \cup \mathcal P_{-\infty} \cup \mathcal P'$. We show how to combine parts of $\overline{\mathcal P}$
into constant-size parts that are also balanced between $\opt$ and $\loc$. Since merging parts does not destroy the property of two centres lying together, we will still have that pairs
of centres in $\cents$ appear together and any pair of centres in $\net$ lie in different parts with probability at most $\eps$.

For any subset $A \subseteq \opt \cup \loc$, let $\mu(A) = |A \cap \opt| - |A \cap \loc|$ denote the {\em imbalance} of $A$.

\begin{lemma} \label{lem:grouping}
Let $Y \geq 1$ be an integer and $\mathcal A$ a collection of disjoint, nonempty subsets of $\opt \cup \loc$ such that $\sum_{A \in \mathcal A} \mu(A) = 0$.
If $|A| \leq Y$ for each $A$ then there is 
some nonempty $\mathcal B \subseteq \mathcal A$ where $|\mathcal B| \leq 2Y^3$ such that $\sum_{B \in \mathcal B} \mu(B) = 0$.
\end{lemma}
\begin{proof}
If $\mu(A) = 0$ for some $A \in \mathcal A$ then simply let $\mathcal B = \{A\}$.
Otherwise, note $|\mu(A)| \leq |A| \leq Y$ for each $A \in \mathcal A$ and partition $\mathcal A$
into sets $\mathcal A_x := \{A \in \mathcal A : \mu(A) = x\}$ for $x \in \{1, \ldots, Y\} \cup \{-1, \ldots, -Y\}$.
We have
\[ \sum_{x=1}^Y |\mathcal A_{-x}| \leq \sum_{x=1}^Y |\mathcal A_{-x}| \cdot x = \sum_{x=1}^Y |\mathcal A_x| \cdot x \leq Y \cdot \sum_{x=1}^Y |\mathcal A_x|, \]
so if $\sum_{x=1}^Y |A_x| \leq Y^2$ then $|\mathcal A| \leq 2Y^3$ and we can take $\mathcal B = \mathcal A$. Similarly
if $\sum_{x=1}^Y |\mathcal A_{-x}| \leq Y^2$ then we can take $\mathcal B = \mathcal A$.

Finally, we are left with the case that $\sum_{x=1}^Y |\mathcal A_x|$ and $\sum_{x=1}^Y |\mathcal A_{-x}|$ both exceed $Y^2$. By the pigeonhole principle,
there are values $1 \leq x, y \leq Y$ such that $|\mathcal A_x| \geq Y$ and $|\mathcal A_{-y}| \geq Y$. In this case, we take $\mathcal B$ to be any $y$ sets
from $\mathcal A_x$ plus any $x$ sets from $\mathcal A_{-y}$.
Then $|\mathcal B| = x + y \leq 2Y$ and $\mathcal B$ is balanced, which is what we needed to show.
\end{proof}

To complete the partitioning we iteratively apply Lemma \ref{lem:grouping} to $\overline{\mathcal P}$
with $Y = 2\cdot (2d)^{2d} \cdot \eps^{-9d/\eps}$ where we note $|P| \leq Y$ by Lemma \ref{lem:size}.
Each iteration, we find some nonempty $\mathcal Q \subseteq \overline{\mathcal P}$
such that $\sum_{Q \in \mathcal Q} \mu(Q) = 0$. Remove $\mathcal Q$ from $\overline{\mathcal P}$ and repeat until all sets 
from $\overline{\mathcal P}$ have been removed.
Each balanced part obtained is the union of at most $2 \cdot \left(2\cdot (2d)^{2d} \cdot \eps^{-9d/\eps}\right)^3$ different parts in $\overline{\mathcal P}$, meaning each
part has size at most $2 \cdot \left(2\cdot (2d)^{2d} \cdot \eps^{-9d/\eps}\right)^4 = 32 \cdot (2d)^{8d} \cdot \eps^{-36 \cdot d/\eps}$.


\section{Proof of Theorem \ref{theo:main} and Extension to \lqnorm}\label{sec:doubling}
In this section we show how to extended the analysis presented in the previous sections for $\RR^d$ to
prove Theorem \ref{theo:main}. For doubling metrics we use $\rho=\swaps=32\cdot d^{16d}\cdot\eps^{-256d/\eps}$ and consider
the $\rho$-swap local search. First we argue why we get a PTAS when 
the metric $\delta(\cdot,\cdot)$ is a doubling metric. 

\subsection{Extending to Doubling Metrics}\label{sec:doubling1}
Let $(V,\delta)$ be a doubling metric space with doubling dimension $d$ (so each ball of radius $2r$ can be
covered with $2^d$ balls of radius $r$ in $V$).
Note that the only places in the analysis where we used the properties of metric being $\RR^d$ was in
the proof of Theorem \ref{thm:partition} and in particular in Step 2 where we cut out cells from bands
and then later in the proof of Lemma \ref{lem:size}. The rest of the proof remains unchanged.

Given a metric $(V,\delta)$, the aspect ratio, denoted by $\Delta$, is the ratio of the largest distance to the smallest 
non-zero distance in the metric: $\Delta = \frac{\max_{u,v \in V}\del(u,v)}{\min_{u,v \in V, u \ne v}\del(u,v)}$. 
Talwar~\cite{Talwar04} gave a hierarchical decomposition of a doubling metric using an
algorithm similar to one by Fakcharoenphol et al.~\cite{FRT03}.
It assumes $\min_{u,v \in V, u \ne v} \del(u,v) = 1$, which can be accomplished by scaling the distances.
So, $\Delta$ is the maximum distance between two points in the metric.

\begin{theorem}\cite{Talwar04}\label{thm:talwar}
Suppose $\min_{u,v \in V : u \ne v} \del(u,v) = 1$.
There is a randomized hierarchical decomposition of $V$, which is a sequence of partitions $\calP_0$, $\calP_1$, $\dots, \calP_h$, 
where $\calP_{i-1}$ is a refinement of $\calP_i$, $\calP_h=\{V\}$, and $\calP_0=\{\{v\}\}_{v\in V}$. The decomposition 
has the following properties:
\begin{enumerate}
\item $\calP_0$ corresponds to the leaves and $\calP_h$ corresponds to the root of the split-tree $T$, and the height of
 $T$ is $h=\varphi+2$, where $\varphi = \log \Delta$ and $\Delta$ is the aspect ratio of metric.
\item For each level $i$ and each $S \in \calP_i$, $S$ has diameter at most $2^{i+1}$.
\item The branching factor $b$ of $T$ is at most $12^d$.
\item For any $u,v \in V$, the probability that they are in different sets corresponding to nodes in level $i$ of $T$ is at most
 $5d\cdot 2^{-i} \cdot \del(u,v)$. 
\end{enumerate}
\end{theorem}

The proof of randomized partitioning scheme for the case of doubling metrics follows similarly
as in Section \ref{sec:partitioning}. We create bands $B_\ell$ as before.
For each band $B_\ell$ we cut out the cells (that define the parts) in the following way. 
First, observe that for all $i\in B_\ell$ we have $\frac{1}{\eps^{\ell b}}\leq D_i< \frac{1}{\eps^{(\ell + 2)b-1}}$.
Additionally, by Lemma \ref{lem:filter} for all $i\in\floc\cap B_\ell$, $i^*\in\fopt\cap B_\ell$ we have 
$\frac{1}{\eps^{\ell b -1}}\leq \del(i,i^*)$.
So, if we scale all distances between points in $B_\ell$ so the minimum distance is $1$, for all $i \in B_\ell$ we have
$1\leq \eps \cdot D_i< \frac{1}{\eps^{2b}}$ and for any two distinct $i, i^* \in B_\ell$ we have $1 \leq \del(i,i^*)$.

Consider the hierarchical decomposition of the points in $B_\ell$ using Theorem \ref{thm:talwar} in this scaled metric
and consider the clusters defined by the
sets at level $\lambda=\log d + (2b+3)\log\frac{1}{\eps}$. These will define our cells for $B_\ell$,
let's call them $C^\ell_1,\ldots,C^\ell_q$.\footnote{It is easy to see that we really don't need the full hierarchical decomposition as in Theorem \ref{thm:talwar}.
One can simply run one round of the {\em Decompose} algorithm of~\cite{Talwar04} that carves out clusters with
diameter size at most $2^\lambda$ and show that the relevant properties of Theorem \ref{thm:talwar} hold.
For ease of presentation we simply invoke Theorem \ref{thm:talwar}.}
The only two things that we need to do is: a) bound the size of each cell $C^\ell_r$
and b) to show that  Lemma \ref{lem:prob} still holds.
Bounding $|C^\ell_r|$ will be equivalent to Lemma \ref{lem:size}.
We then we apply Step 3 to fix $\cents$ as before, this will at most double the size of each cell.
The rest of the proof of Theorem \ref{thm:partition} remains unchanged
as it does not rely on the geometry of $\RR^d$.

To bound $|C^\ell_i|$ we use property 3 of Theorem \ref{thm:talwar} and the fact that the bottom level $\calP_0$ consists of singleton sets:

\begin{eqnarray*}
  |C^\ell_i| &\leq & 12^{d \cdot \lambda}\\
  &=&  2^{(\log 12)d \cdot \left[(2b+3)\log\frac{1}{\eps} + \log d\right]}\\
  &\leq & \eps^{-64d/\eps}\cdot d^{4d},
\end{eqnarray*}
where we used the definition  of $b$ in the last inequality.
Using similar argument as in the case of $\RR^d$, we can bound the size of each part by $4Y^4$
where $Y\leq 2|C^\ell_i|$, which implies an upper bound of $32\cdot d^{16d}\cdot\eps^{-256d/\eps}$.

As for proving equivalent version of Lemma \ref{lem:prob}, it is enough
to show that  $\Pr[i, i^* {\rm ~lie~in~different~} C^\ell_{r}]$ given
that they are in the same band $B_\ell$ is at most $\eps/4$
since the rest of that proof is the same.
For such pair $i,i^*$, note that $\del(i,i^*)\leq \eps^{-1}\cdot D_i \leq \eps^{-2b-1}$.
Using property 4 of Theorem \ref{thm:talwar}:

\begin{eqnarray*}
  \Pr[i, i^* {\rm ~lie~in~different~} C^\ell_{\bf r} ~|~ i, i^* {\rm ~lie~in~the~same~} B_\ell]
  &\leq& 5d\frac{\del(i,i^*)}{2^\lambda}\\
  &\leq& 5d\frac{\eps^{-2b - 1}}{d\eps^{-2b-3}}\\
  &=& 5\eps^{2}\\
  &\leq& \eps/4.
\end{eqnarray*}

The rest of the proof of Theorem \ref{thm:partition} remains the same as in Subsection \ref{sec:balancing}.

\subsection{Extending to \lqnorm} \label{sec:lpq}
It is fairly straightfoward to see that all the calculations in the proof of Theorem \ref{thm:euclid}
where we bound $\Delta^P_j$ will hold with similar bounds if the objective function is sum of distances (instead of squares
of distances), that is $\ell_1$-clustering (i.e. \kmed).
For the case of $\ell^q_q$ (with $q>2$) it is also possible to verify that whenever we have a term
of the form $O(\eps)(c^*_j+c_j)$ when upper bounding $\Delta^P_j$ in the analysis of \kmeans
(e.g. when $j$ is lucky, or long, or good) then 
the equivalent term for $\ell^q_q$ will be $O(\eps 2^q)(\del^q_j+{\del^*_j}^q)$.
For the case when $j$ is bad we still get an upper bound of $\Delta^P_j\leq O(2^q(\del^q_j+{\del^*_j}^q))$.
One can use the crude bound of $(x+\eps\cdot y)^q \leq x^q + 2^q\eps\cdot\max\{x,y\}$ for all $x,y\geq 0$
to help bound terms like $(d_j + \eps \cdot D_{\sigma(j)})^q$.
We skip the straightforward (but tedious) details. 

This means we get a $(1+O(\eps))$-approximation for $k$-clustering when the objective function is measured according to the
$\ell^q_q$-norm for any fixed $q\geq 1$ by considering swaps of size $\rho = d^{O(d)} \cdot (2^q/\eps)^{-d\cdot 2^q/\eps}$
in the local search algorithm.


\section{Extending to the \ufl and \gkm problems} \label{sec:ufl}
Recall that the setting of \ufl is similar to \kmed except we are given opening costs for each $i \in \fa$ instead of a cardinality bound $k$.
A feasible solution is any nonempty $\loc \subseteq \fa$ and its cost is
\[ {\rm cost}_{{\bf ufl}}(\loc) = \sum_{j \in \cX} \del(j, \loc) + \sum_{i \in \loc} f_i.\]
The standard local search algorithm is the following. Let $\rho$ be defined as before.

\begin{algorithm*}[t]
 \caption{$\rho$-Swap Local Search for \ufl} \label{alg:local_ufl}
\begin{algorithmic}
\State Let $\loc \leftarrow \fa$
\While{$\exists$ set $\emptyset \subsetneq \loc' \subseteq\fa$, $|\loc - \loc'|, |\loc' - \loc| \leq \rho$ s.t. 
$\cost_{\bf ufl}(\loc')<\cost_{\bf ufl}(\loc)$}
\State $\loc\leftarrow \loc'$
\EndWhile
\State \Return $\loc$
\end{algorithmic}
\end{algorithm*}

We will show every locally optimum solution has cost at most $(1 + O(\eps)) \cdot OPT$ using at most $2\cdot |\fa|$ test swaps.
Thus, replacing the cost condition in Algorithm \ref{alg:local} 
with $\cost_{\bf ufl}(\loc') \leq \left(1 - \frac{\eps}{2\cdot |\fa|}\right) \cdot \cost_{\bf ufl}(\loc)$
is a polynomial-time variant that is also a PTAS.

Let $\loc$ be a locally optimum solution and $\opt$ a global optimum. We may assume, 
by duplicating points if necessary (for the analysis only), that $\loc \cap \opt = \emptyset$,
thus $|\loc| + |\opt| \leq 2 \cdot |\fa|$ (where $|\fa|$ refers to the size of the original set of facilities before duplication).
Our analysis proceeds in a manner that is nearly identical to our approach for \kmeans.

In particular, we prove the following. It is identical to Theorem \ref{thm:partition} 
in every way except the first point does not require $P \cap \opt$ and $P \cap \loc$ to have equal size.
\begin{theorem}\label{thm:partition_ufl}
There is a randomized algorithm that samples a partitioning $\pi$ of $\opt \cup \loc$ such that:
\begin{itemize}
\item For each part $P \in \pi$, $|P \cap \opt|, |P \cap \loc| \leq \swaps$.
\item For each part $P \in \pi$, $\loc \triangle P$ includes at least one centre from every pair in $\cents$.
\item For each $(i^*, i) \in \net$, $\Pr[i, i^* {\rm ~lie~in~different~parts~of~}\pi] \leq \eps$.
\end{itemize}
\end{theorem}
The proof is identical to Theorem \ref{thm:partition}, except we do not to the ``balancing'' step in Section \ref{sec:balancing}.
Indeed this was the only step in the proof of Theorem \ref{thm:partition} that required $|\loc| = |\opt|$.

We now complete the proof of Theorem \ref{thm:ufl}.\\
\begin{proof}
Let $\costl_j = \del(j, \loc)$ and $\costo_j = \del(j, \opt)$ for each $j \in \cX$.
Sample a partition $\pi$ of $\loc \cup \opt$ as per Theorem \ref{thm:partition_ufl}. For each part $P \in \pi$ and each $j \in \cX$ we let $\Delta^P_j$ denote $\delta(j, \loc \triangle P) - \delta(j, \loc)$.
Using the same bounds considered in the proof of Theorem \ref{thm:euclid} we have
\[ \mathbf{E}_{\pi}\left[\sum_{P \in \pi} \Delta^p_j\right] \leq (1 + O(\eps)) \cdot \costo_j - (1 - O(\eps)) \costl_j. \]

As each $i \in \loc$ is closed exactly once and each $i^* \in \opt$ is opened exactly once, the fact that $\loc$ is locally optimal and fact that each $P \in \pi$ has $|P \cap \loc|, |P \cap \opt| \leq \rho$
shows
\begin{eqnarray*}
0 & \leq & \mathbf{E}_{\pi}\left[\sum_{P \in \pi} \cost_{\bf ufl}(\loc \triangle P) - \cost_{\bf ufl}(\loc)\right] \\
 & \leq & \left(\sum_{j \in \cX} (1 + O(\eps)) \cdot \costo_j - (1 - O(\eps))\costl_j\right)  + \left(\sum_{i^* \in \opt} f_{i^*} - \sum_{i \in \loc} f_i\right).
\end{eqnarray*}
Rearranging shows $\cost_{\bf ufl}(\loc) \leq (1 + O(\eps)) \cost_{\bf ufl}(\opt)$.
\end{proof}

Finally we conclude by analyzing a local search algorithm for \gkm. Here we are given both a cardinality bound $k$ and opening costs for each $i \in \fa$.
The goal is to open some $\loc \subseteq \fa$ with $1 \leq |\loc| \leq k$ to minimize
\[ \cost_{\bf gkm}(\loc) =  \sum_{j \in \cX} \del(j, \loc) + \sum_{i \in \loc} f_i. \]
Note this is the same as $\cost_{\bf ufl}$, we use this slightly different notation to avoid confusion as we are discussing  a different problem now.

\begin{algorithm*}[t]
 \caption{$\rho$-Swap Local Search for \ufl} \label{alg:local_gkm}
\begin{algorithmic}
\State Let $\loc$ be any $k$ centres from $\loc$
\While{$\exists$ set $\loc' \subseteq\fa$, with $|\loc - \loc'|, |\loc' - \loc| \leq \rho$  and $1 \leq |\loc'| \leq k$ s.t.
$\cost_{\bf gkm}(\loc')<\cost_{\bf gkm}(\loc)$}
\State $\loc\leftarrow \loc'$
\EndWhile
\State \Return $\loc$
\end{algorithmic}
\end{algorithm*}

We can prove Theorem \ref{thm:gkm} very easily.

\begin{proof}
Let $\loc$ and $\opt$ denote a local optimum and a global optimum (respectively). By adding artificial points $i$ with $f_i = 0$ that are extremely far away from $\cX$, we may assume $\loc = \opt$.
Then we simply use our original partitioning result, Theorem \ref{thm:partition}, to define test swaps. Noting that each $i \in \loc$ is closed exactly once and each $i^* \in \opt$ is opened exactly once
over the swaps induced by a partition $\pi$, we proceed in the same way as above in the proof of Theorem \ref{thm:ufl} and see
\[ \cost_{\bf gkm}(\loc) \leq (1 + O(\eps)) \cdot \cost_{\bf gkm}(\opt). \]
\end{proof}

Using almost identical arguments as we did in 
Section \ref{sec:lpq} for extension of \kmeans to $\ell^q_q$-norm in doubling dimensions, 
one can easily extend the results of Theorem
\ref{thm:ufl} and \ref{thm:gkm} to the case where we consider the sum of $q$'th power of distances ($\ell^q_q$-norm).


\section{Conclusion}
We have presented a PTAS for \kmeans in constant-dimensional Euclidean metrics and, more generally, in doubling metrics and when the objective
is to minimizing the $\ell_q^q$-norm of distances between points and their nearest centres for any constant $q \geq 1$.
This is also the first approximation for \kmed in doubling metrics and the first demonstration that local search yields a true PTAS
for \kmed even in the Euclidean plane. Our approach extends to showing local search yields a PTAS for \ufl and \gkm in doubling metrics.

The running time of a single step of the local search algorithm is $O(k^{\rho})$ where 
$\rho = d^{O(d)} \cdot \eps^{-O(d/\eps)}$ for the case of \kmeans when the metric has doubling dimension $d$.
We have not tried to optimize the constants in the $O(\cdot)$ notations in $\rho$.
The dependence on $d$ cannot be improved much under the Exponential Time Hypothesis (\ETH).
For example, if the running time of a single iteration was only $O(k^{{\rm exp}(d^{1-\delta}) \cdot f(\eps)})$ for some constant
$\delta$ then we would have a sub-exponential time $(1+\eps)$-approximation when $d = O(\log n)$. Recall that \kmeans is \apx-hard when $d = \Theta(\log n)$
\cite{ACKS15}, so this would refute the \ETH.

It may still be possible to obtain an EPTAS for any constant dimension $d$. That is, there could be a PTAS with running time of the form $O(g(\eps) \cdot n^{{\rm exp}(d)})$ for some function
$g(\eps)$ (perhaps depending also on $d$). Finally, what is the fastest PTAS we can obtain in the special case of 
the Euclidean plane (i.e. $d = 2$)?
It would be interesting to see if there is
an EPTAS whose running time is linear or near linear in $n$ for any fixed constant $\eps$.


\bibliographystyle{plain}
{
\bibliography{k-means}

\begin{thebibliography}{10}

\bibitem{ADHP09}
Daniel Aloise, Amit Deshpande, Pierre Hansen, and Preyas Popat.
\newblock {NP}-hardness of {Euclidean} sum-of-squares clustering.
\newblock {\em Mach. Learn.}, 75(2):245--248, 2009.

\bibitem{Arora98}
Sanjeev Arora.
\newblock Polynomial time approximation schemes for {Euclidean} traveling
  salesman and other geometric problems.
\newblock {\em J. ACM}, 45(5):753--782, 1998.

\bibitem{ARR98}
Sanjeev Arora, Prabhakar Raghavan, and Satish Rao.
\newblock Approximation schemes for {Euclidean} {K}-medians and related
  problems.
\newblock In {\em Proceedings of the Thirtieth Annual ACM Symposium on Theory
  of Computing (STOC '98)}, pages 106--113. ACM, 1998.

\bibitem{AMR11}
David Arthur, Bodo Manthey, and Heiko R\"{o}glin.
\newblock Smoothed analysis of the {K}-means method.
\newblock {\em J. ACM}, 58(5):19:1--19:31, 2011.

\bibitem{AV06}
David Arthur and Sergei Vassilvitskii.
\newblock How slow is the {K}-means method?
\newblock In {\em Proceedings of the Twenty-second Annual Symposium on
  Computational Geometry (SoCG '06)}, pages 144--153. ACM, 2006.

\bibitem{AV07}
David Arthur and Sergei Vassilvitskii.
\newblock {K}-means++: {The} advantages of careful seeding.
\newblock In {\em Proceedings of the Eighteenth Annual ACM-SIAM Symposium on
  Discrete Algorithms (SODA '07)}, pages 1027--1035. SIAM, 2007.

\bibitem{AGKMMP01}
Vijay Arya, Naveen Garg, Rohit Khandekar, Adam Meyerson, Kamesh Munagala, and
  Vinayaka Pandit.
\newblock Local search heuristic for {K}-median and facility location problems.
\newblock In {\em Proceedings of the Thirty-third Annual ACM Symposium on
  Theory of Computing (STOC '01)}, pages 21--29. ACM, 2001.

\bibitem{AGKMMP04}
Vijay Arya, Naveen Garg, Rohit Khandekar, Adam Meyerson, Kamesh Munagala, and
  Vinayaka Pandit.
\newblock Local search heuristics for {K}-median and facility location
  problems.
\newblock {\em {SIAM} J. Comput.}, 33(3):544--562, 2004.

\bibitem{ABS10}
Pranjal Awasthi, Avrim Blum, and Or~Sheffet.
\newblock Stability yields a {PTAS} for {K}-median and {K}-means clustering.
\newblock In {\em Proceedings of the 2010 IEEE 51st Annual Symposium on
  Foundations of Computer Science (FOCS '10)}, pages 309--318. IEEE Computer
  Society, 2010.

\bibitem{ACKS15}
Pranjal Awasthi, Moses Charikar, Ravishankar Krishnaswamy, and Ali~Kemal Sinop.
\newblock {The Hardness of Approximation of {Euclidean} {K}-Means}.
\newblock In {\em Proceedings of 31st International Symposium on Computational
  Geometry (SoCG '15)}, Leibniz International Proceedings in Informatics
  (LIPIcs), pages 754--767, 2015.

\bibitem{BHI02}
Mihai B\={a}doiu, Sariel Har-Peled, and Piotr Indyk.
\newblock Approximate clustering via {Coresets}.
\newblock In {\em Proceedings of the Thiry-fourth Annual ACM Symposium on
  Theory of Computing (STOC '02)}, pages 250--257. ACM, 2002.

\bibitem{BV16}
Sayan Bandyapadhyay and Kasturi Varadarajan.
\newblock On variants of {K}-means clustering.
\newblock In {\em Proceedings of the 32nd International Symposium on
  Computational Geometry (SoCG '16)}, Leibniz International Proceedings in
  Informatics (LIPIcs), 2016.

\bibitem{BLSS16}
Johannes Bl{\"{o}}mer, Christiane Lammersen, Melanie Schmidt, and Christian
  Sohler.
\newblock Theoretical analysis of the {K}-means algorithm - {A} survey.
\newblock {\em CoRR}, abs/1602.08254, 2016.

\bibitem{BPRST15}
Jaroslaw Byrka, Thomas Pensyl, Bartosz Rybicki, Aravind Srinivasan, and Khoa
  Trinh.
\newblock An improved approximation for {K}-median, and positive correlation in
  budgeted optimization.
\newblock In {\em Proceedings of the Twenty-Sixth Annual ACM-SIAM Symposium on
  Discrete Algorithms (SODA '15)}, pages 737--756. SIAM, 2015.

\bibitem{Chen09}
Ke~Chen.
\newblock On {Coresets} for {K}-median and {K}-means clustering in metric and
  {Euclidean} spaces and their applications.
\newblock {\em SIAM J. Comput.}, 39:923--947, 2009.

\bibitem{AKM16B}
Vincent Cohen{-}Addad, Philip~N. Klein, and Claire Mathieu.
\newblock Local search yields approximation schemes for k-means and k-median in
  euclidean and minor-free metrics.
\newblock In {\em {IEEE} 57th Annual Symposium on Foundations of Computer
  Science, {FOCS} 2016, 9-11 October 2016, Hyatt Regency, New Brunswick, New
  Jersey, {USA}}, pages 353--364, 2016.

\bibitem{AKM16}
Vincent Cohen{-}Addad, Philip~N. Klein, and Claire Mathieu.
\newblock The power of local search for clustering.
\newblock {\em CoRR}, abs/1603.09535, 2016.

\bibitem{CM15}
Vincent Cohen-Addad and Claire Mathieu.
\newblock Effectiveness of local search for geometric optimization.
\newblock In {\em Proceedings of the 31st International Symposium on
  Computational Geometry (SoCG '15)}, Leibniz International Proceedings in
  Informatics (LIPIcs), 2015.

\bibitem{CMSOGC15}
Vincent Cohen-Addad and Claire Mathieu.
\newblock Effectiveness of local search for geometric optimization.
\newblock In {\em Proceedings of 31st Annual Symposium on Computational
  Geometry (SOCG '15)}, pages 329--344, 2015.

\bibitem{Dasgupta03a}
Sanjoy Dasgupta.
\newblock How fast is {K}-means?
\newblock In {\em Proceedings of the 16th Annual Conference on Learning Theory
  and 7th Kernel Workshop (COLT/Kernel '03)}, pages 735--735, 2003.

\bibitem{DKKR03}
W.~Fernandez de~la Vega, Marek Karpinski, Claire Kenyon, and Yuval Rabani.
\newblock Approximation schemes for clustering problems.
\newblock In {\em Proceedings of the Thirty-fifth Annual ACM Symposium on
  Theory of Computing (STOC '03)}, pages 50--58. ACM, 2003.

\bibitem{Patrik2005}
Patrik D'haeseleer et~al.
\newblock How does gene expression clustering work?
\newblock {\em Nature biotechnology}, 23(12):1499--1502, 2005.

\bibitem{DFKVV04}
P.~Drineas, A.~Frieze, R.~Kannan, S.~Vempala, and V.~Vinay.
\newblock Clustering large graphs via the singular value decomposition.
\newblock {\em Mach. Learn.}, 56(1-3):9--33, 2004.

\bibitem{FRT03}
Jittat Fakcharoenphol, Satish Rao, and Kunal Talwar.
\newblock A tight bound on approximating arbitrary metrics by tree metrics.
\newblock In {\em Proceedings of the Thirty-fifth Annual ACM Symposium on
  Theory of Computing (STOC '03)}, pages 448--455. ACM, 2003.

\bibitem{FMS07}
Dan Feldman, Morteza Monemizadeh, and Christian Sohler.
\newblock A {PTAS} for {K}-means clustering based on weak {Coresets}.
\newblock In {\em Proceedings of the Twenty-third Annual Symposium on
  Computational Geometry (SoCG '07)}, SoCG '07, pages 11--18. ACM, 2007.

\bibitem{FSS13}
Dan Feldman, Melanie Schmidt, and Christian Sohler.
\newblock Turning big data into tiny data: Constant-size {Coresets} for
  {K}-means, {PCA} and projective clustering.
\newblock In {\em Proceedings of the Twenty-Fourth Annual ACM-SIAM Symposium on
  Discrete Algorithms (SODA '13)}, pages 1434--1453. SIAM, 2013.

\bibitem{FRS16B}
Zachary Friggstad, Mohsen Rezapour, and Mohammad~R. Salavatipour.
\newblock Local search yields a ptas for k-means in doubling metrics.
\newblock In {\em {IEEE} 57th Annual Symposium on Foundations of Computer
  Science, {FOCS} 2016, 9-11 October 2016, Hyatt Regency, New Brunswick, New
  Jersey, {USA}}, pages 365--374, 2016.

\bibitem{FRS16}
Zachary Friggstad, Mohsen Rezapour, and Mohammad~R. Salavatipour.
\newblock Local search yields a {PTAS} for k-means in doubling metrics.
\newblock {\em CoRR}, abs/1603.08976, 2016.

\bibitem{GT08}
A.~Gupta and T.~Tangwongsan.
\newblock Simpler analyses of local search algorithms for facility location.
\newblock {\em CoRR, abs/0809.2554}, 2008.

\bibitem{HK05}
Sariel Har-Peled and Akash Kushal.
\newblock Smaller {Coresets} for {K}-median and {K}-means clustering.
\newblock In {\em Proceedings of the Twenty-first Annual Symposium on
  Computational Geometry (SoCG '05)}, pages 126--134. ACM, 2005.

\bibitem{HM04}
Sariel Har-Peled and Soham Mazumdar.
\newblock On {Coresets} for {K}-means and {K}-median clustering.
\newblock In {\em Proceedings of the Thirty-sixth Annual ACM Symposium on
  Theory of Computing (STOC '04)}, pages 291--300. ACM, 2004.

\bibitem{HS05}
Sariel Har-Peled and Bardia Sadri.
\newblock How fast is the {K}-means method?
\newblock {\em Algorithmica}, 41(3):185--202, 2005.

\bibitem{Hofree2013}
Matan Hofree, John~P Shen, Hannah Carter, Andrew Gross, and Trey Ideker.
\newblock Network-based stratification of tumor mutations.
\newblock {\em Nature methods}, 10(11):1108--1115, 2013.

\bibitem{Inaba1994}
Mary Inaba, Naoki Katoh, and Hiroshi Imai.
\newblock Applications of weighted {Voronoi} diagrams and randomization to
  variance-based {K}-clustering.
\newblock In {\em Proceedings of the tenth annual Symposium on Computational
  Geometry (SoCG '94)}, pages 332--339. ACM, 1994.

\bibitem{Jain10}
Anil~K. Jain.
\newblock Data clustering: 50 years beyond {K}-means.
\newblock {\em Pattern Recogn. Lett.}, 31(8):651--666, 2010.

\bibitem{KMNPSW04}
Tapas Kanungo, David~M. Mount, Nathan~S. Netanyahu, Christine~D. Piatko, Ruth
  Silverman, and Angela~Y. Wu.
\newblock A local search approximation algorithm for {K}-means clustering.
\newblock {\em Comput. Geom. Theory Appl.}, 28(2-3):89--112, 2004.

\bibitem{KR99}
Stavros~G. Kolliopoulos and Satish Rao.
\newblock A nearly linear-time approximation scheme for the {Euclidean}
  {Kappa}-median problem.
\newblock In {\em Proceedings of the 7th Annual European Symposium on
  Algorithms (ESA '99)}, pages 378--389. Springer-Verlag, 1999.

\bibitem{KK10}
Amit Kumar and Ravindran Kannan.
\newblock Clustering with spectral norm and the {K}-means algorithm.
\newblock In {\em Proceedings of the 2010 IEEE 51st Annual Symposium on
  Foundations of Computer Science (FOCS '10)}, pages 299--308. IEEE Computer
  Society, 2010.

\bibitem{KSS04}
Amit Kumar, Yogish Sabharwal, and Sandeep Sen.
\newblock A simple linear time (1+ $\epsilon$) -approximation algorithm for
  {K}-means clustering in any dimensions.
\newblock In {\em Proceedings of the 45th Annual IEEE Symposium on Foundations
  of Computer Science (FOCS '04)}, pages 454--462. IEEE Computer Society, 2004.

\bibitem{KSS10}
Amit Kumar, Yogish Sabharwal, and Sandeep Sen.
\newblock Linear-time approximation schemes for clustering problems in any
  dimensions.
\newblock {\em J. ACM}, 57(2):5:1--5:32, 2010.

\bibitem{LS10}
Michael Langberg and Leonard~J. Schulman.
\newblock Universal $\epsilon$-approximators for integrals.
\newblock In {\em Proceedings of the Twenty-first Annual ACM-SIAM Symposium on
  Discrete Algorithms (SODA '10)}, pages 598--607. SIAM, 2010.

\bibitem{L11}
Shi Li.
\newblock A 1.488-approximation for the uncapacitated facility location
  problem.
\newblock In {\em Proceedings of the 38th Annual International Colloquium on
  Automata, Languages and Programming (ICALP '11)}, pages 45--58, 2011.

\bibitem{LS13}
Shi Li and Ola Svensson.
\newblock Approximating {K}-median via pseudo-approximation.
\newblock In {\em Proceedings of the Forty-fifth Annual ACM Symposium on Theory
  of Computing (STOC '13)}, pages 901--910. ACM, 2013.

\bibitem{Lloyd82}
S.~Lloyd.
\newblock Least squares quantization in pcm.
\newblock {\em IEEE Trans. Inf. Theor.}, 28(2):129--137, 2006.

\bibitem{MNV09}
Meena Mahajan, Prajakta Nimbhorkar, and Kasturi Varadarajan.
\newblock The planar {K}-means problem is {NP}-hard.
\newblock In {\em Proceedings of the 3rd International Workshop on Algorithms
  and Computation (WALCOM '09)}, pages 274--285. Springer-Verlag, 2009.

\bibitem{Matousek00}
Jir{\i} Matou{\v{s}}ek.
\newblock On approximate geometric k-clustering.
\newblock {\em Discrete \& Computational Geometry}, 24(1):61--84, 2000.

\bibitem{OR02}
Rafail Ostrovsky and Yuval Rabani.
\newblock {Polynomial}-{Time} {Approximation} {Schemes} for geometric min-sum
  median clustering.
\newblock {\em J. ACM}, 49(2):139--156, 2002.

\bibitem{ORSS12}
Rafail Ostrovsky, Yuval Rabani, Leonard~J. Schulman, and Chaitanya Swamy.
\newblock The effectiveness of {Lloyd}-type methods for the {K}-means problem.
\newblock {\em J. ACM}, 59(6):28:1--28:22, 2013.

\bibitem{Talwar04}
Kunal Talwar.
\newblock Bypassing the embedding: Algorithms for low dimensional metrics.
\newblock In {\em Proceedings of the Thirty-sixth Annual ACM Symposium on
  Theory of Computing (STOC '04)}, pages 281--290. ACM, 2004.

\bibitem{Vattani09}
Andrea Vattani.
\newblock The hardness of {K}-means clustering in the plane.
\newblock {\em Manuscript}, 2009.

\bibitem{Vattani11}
Andrea Vattani.
\newblock {K}-means requires exponentially many iterations even in the plane.
\newblock {\em Discrete Comput. Geom.}, 45(4):596--616, 2011.

\end{thebibliography}
}

\end{document}